\documentclass{article}

\usepackage{microtype}
\usepackage{graphicx}
\usepackage{subcaption}
\usepackage{booktabs} 
\usepackage{bm}
\usepackage{hyperref}

\usepackage[preprint]{icml2026}

\usepackage{amsmath}
\usepackage{amssymb}
\usepackage{mathtools}
\usepackage{amsthm}

\usepackage[capitalize,noabbrev]{cleveref}

\theoremstyle{plain}
\newtheorem{theorem}{Theorem}[section]
\newtheorem{proposition}[theorem]{Proposition}
\newtheorem{lemma}[theorem]{Lemma}

\theoremstyle{definition}
\newtheorem{definition}[theorem]{Definition}

\theoremstyle{remark}

\newcommand{\subfiglabel}[1]{\refstepcounter{subfigure}\label{#1}}
\makeatletter
\renewcommand{\p@subfigure}{\thefigure} 
\makeatother
\usepackage{hyperref}

\newcommand{\subfigtarget}[1]{\hypertarget{#1}{}}

\makeatletter
\renewcommand{\subfiglabel}[1]{%
  \refstepcounter{subfigure}%
  \begingroup
    \protected@edef\@currentHref{#1}
    \label{#1}%
  \endgroup
}
\makeatother

\usepackage[textsize=tiny]{todonotes}

\icmltitlerunning{Efficient Learned Image Compression without Entropy Coding}

\begin{document}

\twocolumn[
  \icmltitle{Efficient Learned Image Compression without Entropy Coding}



  \icmlsetsymbol{equal}{*}

  \begin{icmlauthorlist}
    \icmlauthor{Hao Cao}{thu}
    \icmlauthor{Wenqi Guo}{thu2}
    \icmlauthor{Zhijin Qin}{thu,thu3}
    \icmlauthor{Jungong Han}{thu2,thu4}
  \end{icmlauthorlist}

  \icmlaffiliation{thu}{Department of Electronic Engineering, Tsinghua University}
  \icmlaffiliation{thu2}{Department of Automation, Tsinghua University}
  \icmlaffiliation{thu3}{State Key Laboratory of Space Network and Communications}
  \icmlaffiliation{thu4}{Beijing National Research Center for Information Science and Technology}

  \icmlcorrespondingauthor{Wenqi Guo}{wenqiguo@mail.tsinghua.edu.cn}

  \icmlkeywords{Image compression}

  \vskip 0.3in
]



\printAffiliationsAndNotice{}  


\begin{abstract}
Entropy coding is widely used in typical learned image compression (LIC) that converts latents into a compact bitstream.
However, entropy coding is typically sequential and becomes the coding latency bottleneck.
To overcome it, we present \textbf{E}ntropy-Coding \textbf{F}ree \textbf{L}earned \textbf{I}mage \textbf{C}ompression (EF-LIC), a multi-rate framework that generates compact representation by removing statistical and correlation redundancy with low coding latency.
First, we introduce unconstrained vector quantization and prove that its index distribution approaches the maximum-entropy bound, yielding minimal statistical redundancy.
Second, we propose a context-conditioned autoregressive transform that directly reparameterizes the latents to reduce inter-dependency.
Theoretical analysis shows that EF-LIC can remove correlation redundancy as effectively as typical LIC with entropy coding, leading to comparable compression performance.
Experiments show EF-LIC achieves up to 67.86\% bitrate reduction over MS-ILLM on Kodak with LPIPS.
Ablation studies further show EF-LIC matches the compression performance of its entropy-coding based variant while achieving over $3\times$ faster encoding and $5\times$ faster decoding.
\end{abstract}

\section{Introduction}

\begin{figure}[t]
  \centering
  \begin{subfigure}[c]{0.95\linewidth}
    \includegraphics[width=0.9999\linewidth]{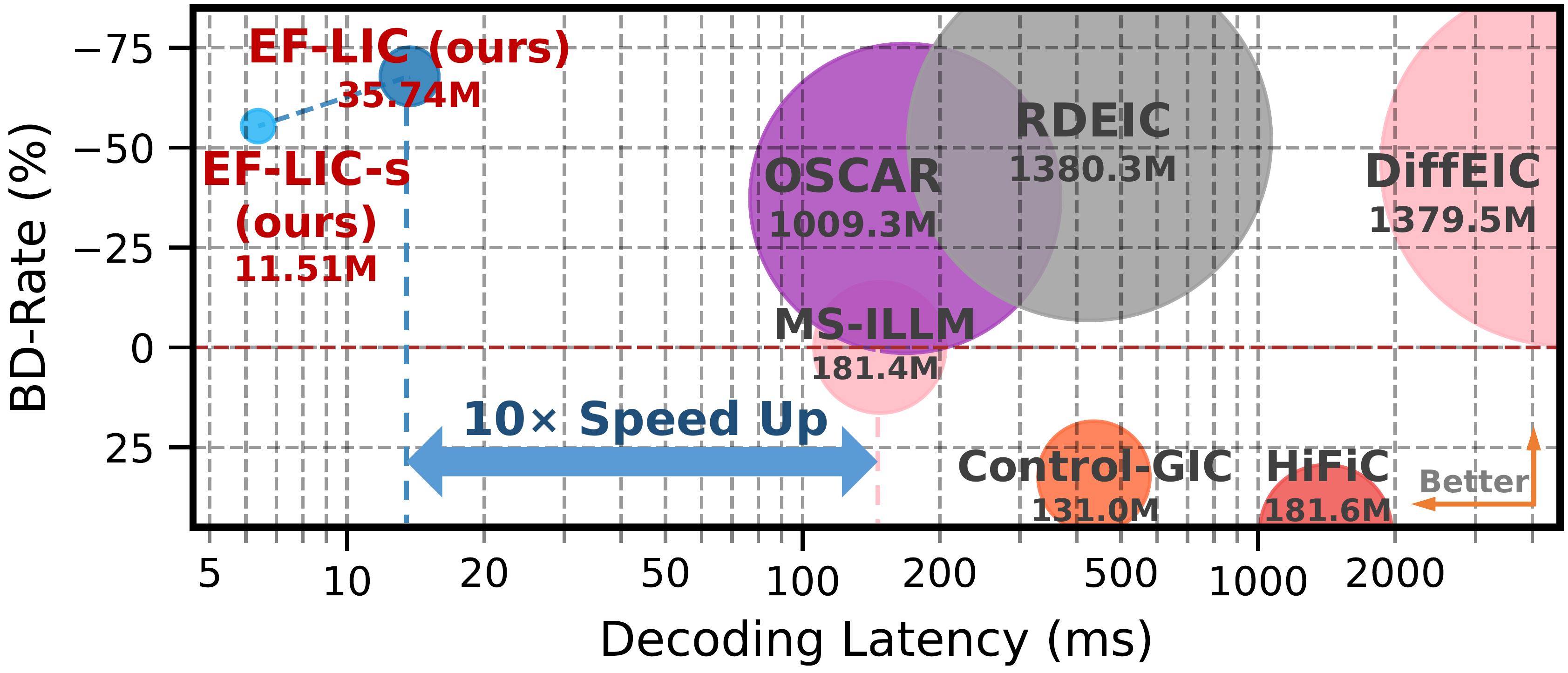}
    \caption{Performance comparison with LIC on LPIPS.}
    \label{fig:bd1}
  \end{subfigure}
  \hfill
  \begin{subfigure}[c]{0.95\linewidth}
    \includegraphics[width=0.9999\linewidth]{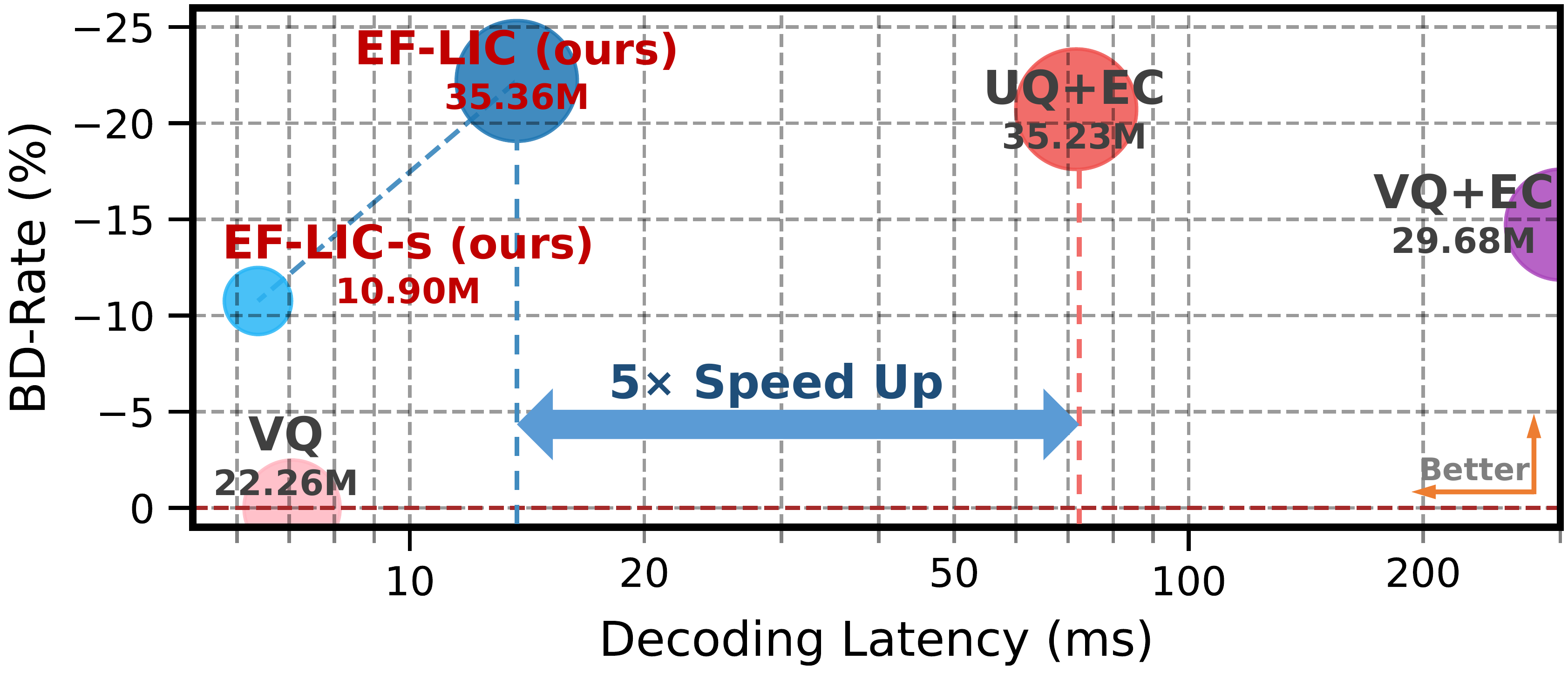}
    \caption{Ablation studies with different variants.}
    \label{fig:bd2}
  \end{subfigure}
  \caption{(a) EF-LIC is the proposed method, which achieves high performance and low decoding latency. EF-LIC-s is its lightweight variant.  (b) Comparison of EF-LIC with its variants. ``UQ+EC'' denotes typical LIC with uniform quantization (UQ), context modeling, and entropy coding. ``VQ'' is the baseline method without inter-latent decorrelation. ``VQ+EC'' denotes context modeling and entropy coding for discrete VQ indices. All of them share the same module structure and distortion metrics. Results are reported on Kodak using LPIPS, evaluated with one NVIDIA A100 GPU.}
  \label{fig:bd}
\end{figure}

Lossy image compression~\cite{jpeg} seeks a compact representation that minimizes bitrate while preserving high quality.
To this end, information theory~\cite{shannon1948mathematical} offers a principled lens that views compression as redundancy removal, where redundancy can be divided into (i) statistical redundancy and (ii) correlation redundancy.
Statistical redundancy arises when the distribution of the quantized latents follows a non-uniform distribution.
In this case, entropy coding \cite{huffman} could assign shorter codewords to more probable latents, reducing the expected number of bits.
Correlation redundancy arises when latents are statistically dependent across positions, making some symbols predictable from their context.
In learned image compression (LIC), a context model~\cite{lic-joint} often implemented via context-conditional autoregressive transform~\cite{lic-he2021checkerboard, lic-elic, hpcm}, captures inter-latent dependency through a conditional distribution.
Entropy coding can then exploit this conditional distribution to reduce correlation redundancy.

Therefore, entropy coding plays a central role in typical LIC, as it converts latents into a compact bitstream.
However, its complex and sequential control flow is hard to parallelize, so entropy coding is often implemented on CPUs and can become the primary bottleneck in end-to-end latency.
For example, range Asymmetric Numeral Systems (rANS)~\cite{rans} can take more than 100\,ms, exceeding the combined runtime of other modules in the LIC pipeline.
Meanwhile, simplifying or removing entropy coding typically incurs a substantial performance degradation.
For instance, Huffman coding~\cite{huffman} is faster but far less efficient than rANS.
Prior LIC methods, such as COIN~\cite{coin} and OSCAR~\cite{oscar}, exclude entropy coding, but they either only achieve the performance of simple codecs such as JPEG~\cite{jpeg} or incur prohibitive inference cost.
These issues motivate a natural question: \textit{How can we perform image compression without entropy coding while preserving high compression efficiency?}

To address this question, we propose \textbf{E}ntropy-coding \textbf{F}ree \textbf{L}earned \textbf{I}mage \textbf{C}ompression (EF-LIC), a multi-rate framework that achieves high compression efficiency with low coding latency.
Following information theory \cite{shannon1948mathematical}, the first challenge is to remove statistical redundancy, which amounts to learning latents whose entropy approaches the maximum.
We introduce unconstrained vector quantization (VQ) \cite{vq-vae}, and prove that the index sequence from VQ exhibits minimal statistical redundancy.
The second challenge is to remove correlation redundancy, which amounts to eliminating repeated information across latents.
To avoid predicting the conditional distribution as typical LIC, we propose representation-domain latent decorrelation, which contains a context-conditioned autoregressive transform to directly reparameterize the latents with reduced correlation.
These two steps are GPU friendly and enable EF-LIC to break the latency--performance trade-off.
We also adopt residual vector quantization (RVQ)~\cite{speech-dac} to enable flexible multi-rate compression.

We evaluate EF-LIC under perceptual metrics, which better reflect visual quality than pixel-wise metrics such as PSNR. We report BD-rate~\cite{bd-rate} to calculate bitrate reduction under the same distortion.
As shown in~\cref{fig:bd1}, EF-LIC achieves 67.86\% bitrate reduction over MS-ILLM~\cite{ms-illm} evaluated with LPIPS~\cite{lpips} on Kodak.
Ablation studies in~\cref{fig:bd2} show that EF-LIC matches the performance of its entropy-coding based variant with the same architecture, while delivering over $5\times$ faster decoding.

Our contributions are summarized as follows.
\begin{itemize}
  \item We propose \textbf{E}ntropy-coding \textbf{F}ree \textbf{L}earned \textbf{I}mage \textbf{C}ompression (EF-LIC), a multi-rate LIC achieving both high compression performance and low latency. It combines unconstrained VQ to produce high-entropy discrete indices and a context-conditional autoregressive transform that reparameterizes the latents.
  \item We provide theoretical analyses that (i) unconstrained VQ produces discrete indices with minimal statistical redundancy as the model approaches the minimum reconstruction distortion, and (ii) the context-conditional autoregressive transform achieves the same compression performance as typical LIC with entropy coding.
  \item Experiments show that EF-LIC both improves compression performance and decreases latency over prior methods. It also matches the compression performance of its entropy-coding-based variant while providing a significant encoding and decoding speedup.
\end{itemize}
\vspace{-2pt}
\section{Related Work}
\vspace{-2pt}
\paragraph{Learned Image Compression.}
Pioneering work~\cite{lic-balle2018variational} introduces variational autoencoders~\cite{vae} for LIC. Subsequent studies outperformed traditional codecs such as JPEG~\cite{jpeg} and VVC~\cite{vtm}. Several studies improve transform coding~\cite{lic-attention, lic-tcm, lalic}, while others advance context modeling~\cite{lic-lu2025learned, hpcm}. 
Notably, a context model translates inter-latent dependency into a conditional probability exploited by entropy coding to reduce the expected bitstream length.
An early study~\cite{lic-balle2018variational} introduces hyperpriors to model the conditional distributions of the latents.
Later, autoregressive models~\cite{lic-joint, lic-attention} partition the latents into groups and model inter-group dependency. 
Afterwards, more studies improve either the grouping strategy \cite{lic-he2021checkerboard, lic-elic, dcvc-dc, hpcm} or the model capacity \cite{lic-mlic, lic-lu2025learned} to reduce inter-latent dependency, but still rely on entropy coding for bitstream generation.
\vspace{-2pt}
\paragraph{Generative Image Compression.}
Early approaches are mainly optimized for pixel-level distortion (e.g., PSNR). However, these objectives often correlate poorly with human perception~\cite{perception}.
Motivated by this mismatch, several works~\cite{gic-1, lic-poelic} aim to better align LIC optimization with visual quality.
HiFiC~\cite{hific} leverages GANs~\cite{gan} to improve the visual quality of reconstructions.
MS-ILLM~\cite{ms-illm} further refines the discriminator architecture to improve distributional alignment between reconstructions and natural images.
Building on them, subsequent studies explore VQ-GAN~\cite{vq-gan} for LIC, achieving high visual quality at extremely low bitrates~\cite{vq-mao2024extreme, vq-gic, cgic, vq-DLF}.
Diffusion-based generative compression has been explored in several works~\cite{ddpm, diffusion-perco, diffusion-StableCodec,diffusion-xue2025one, rdeic} to reconstruct high-quality images. However, the cost of diffusion inference limits practical deployment.

\begin{figure*}[t]
  \centering
  \subfigtarget{fig:framework-a}%
  \subfigtarget{fig:framework-b}%
  \subfigtarget{fig:framework-c}%
  \includegraphics[width=0.975\linewidth]{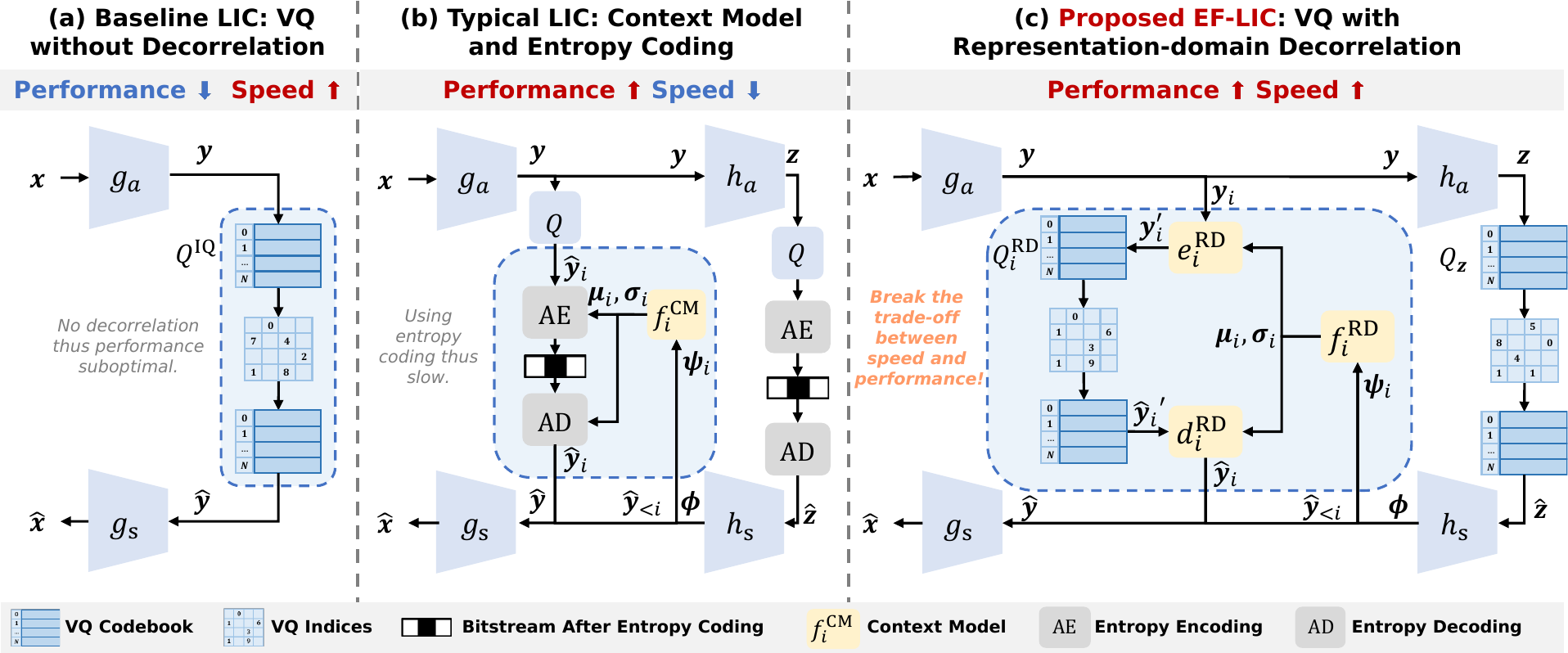}
  \caption{(a) Left: a VQ-only baseline that is fast but less efficient due to missing inter-latent decorrelation. (b) Middle: a typical entropy-coded LIC pipeline, where the context model $\bm{f}^{\mathrm{CM}}$ outputs conditional probabilities for AE and AD. (c) Right: the proposed EF-LIC, which applies a context-conditional transform to produce low-correlation latents and uses unconstrained VQ to remove redundancy.}
  \label{fig:framework}
  \setcounter{subfigure}{0}
  \subfiglabel{fig:framework-a}
  \subfiglabel{fig:framework-b}
  \subfiglabel{fig:framework-c}
\end{figure*}

\vspace{-2pt}
\paragraph{Image Compression without Entropy Coding.}
There have been works of LIC without entropy coding \cite{lic-full}.
COIN~\cite{coin} adopts implicit neural representations without introducing entropy coding, but its compression performance remains comparable only to JPEG-level codecs.
OSCAR~\cite{oscar} engages diffusion models while excluding entropy coding, but incurs prohibitive computational cost.
Another line of work uses vector quantization (VQ)~\cite{vq-vae} to map continuous latents to discrete code indices, so the compressed representation reduces to an index sequence. Nevertheless, they \cite{vq-mao2024extreme,cgic} typically overlook inter-latent dependency, resulting in suboptimal compression efficiency. While some studies \cite{vq-pq-mim,vq-gic,vq-mcquic} reduce inter-latent correlation after vector quantization, they still rely on entropy coding.
\section{Methods}

\subsection{Overview Architecture of EF-LIC}
\label{sec:architecture}

Unlike existing LIC shown in \cref{fig:framework-a,fig:framework-b}, EF-LIC removes redundancy via VQ and  context-conditioned transforms to generate compact representation directly without entropy coding.
As shown in \cref{fig:framework-c}, EF-LIC encodes an input image $\bm{x}\in\mathbb{R}^{3\times H\times W}$ into a latent $\bm{y}=g_a(\bm{x})$ with a downsampling factor of $f_y$. EF-LIC adopts a hyperprior~\cite{lic-balle2018variational} branch to extract side information as $\bm{z}=h_a(\bm{y})$ with a downsampling factor $f_z$. The hyperprior is then quantized as $\hat{\bm{z}}=Q_{\bm{z}}(\bm{z})$ and decoded into a context feature $\bm{\phi}=h_s(\hat{\bm{z}})$ to assist the decorrelation of $\bm{y}$.

Notably, we propose representation-domain decorrelation (RD), which generates a new latent directly from $\bm{y}$  instead of predicting a conditional probability distribution. 
Specifically, the latent $\bm{y}$ is partitioned into $N$ groups as $(\bm{y}_1,\ldots,\bm{y}_N)$.
For the $i$-th group, a reference context $\bm{\psi}_i$ is first constructed from the decoded groups $\hat{\bm{y}}_{<i}$ and the context feature $\bm{\phi}$: $\bm{\psi}_i=\mathrm{concat}\!\left(\hat{\bm{y}}_{<i}, \bm{\phi}\right)$, where $\mathrm{concat}(\cdot,\cdot)$ denotes concatenation and $\hat{\bm{y}}_{<i}=(\hat{\bm{y}}_1,\ldots,\hat{\bm{y}}_{i-1})$.
Subsequently, a context extractor $f_i^{\mathrm{RD}}(\cdot)$ transforms the reference context $\bm{\psi}_i$ into the context parameters $(\bm{\mu}_i,\bm{\sigma}_i)$ as
\begin{equation}
(\bm{\mu}_i, \bm{\sigma}_i)=f^{\mathrm{RD}}_i(\bm{\psi}_i).
\end{equation}
A context-conditional encoder $e^{\mathrm{RD}}_i(\cdot;\cdot)$ reparameterizes the current group $\bm{y}_i$ via an affine projection as:
\begin{equation}
\bm{y}'_i = e^{\mathrm{RD}}_i(\bm{y}_i;\bm{\mu}_i,\bm{\sigma}_i)
= \bm{\sigma}_i^{-1}\odot(\bm{y}_i-\bm{\mu}_i),
\end{equation}
where $\odot$ is elementwise multiplication. 
Then $\bm{y}'_i$ is quantized as $\hat{\bm{y}}'_i = Q_i^{\mathrm{RD}}(\bm{y}'_i)$, where $Q_i^{\mathrm{RD}}(\cdot)$ denotes a group-wise vector quantizer.
The quantized latent $\hat{\bm{y}}_i$ is reconstructed from $\hat{\bm{y}}'_i$ using a context-conditional decoder $d_i^{\mathrm{RD}}(\cdot;\cdot)$ as
\begin{equation}
\hat{\bm{y}}_i = d^{\mathrm{RD}}_i(\hat{\bm{y}}'_i;\bm{\mu}_i,\bm{\sigma}_i)
= \bm{\sigma}_i\odot\hat{\bm{y}}'_i+\bm{\mu}_i.
\end{equation}
Note that $e_i^{\mathrm{RD}}(\cdot;\cdot)$ and $d_i^{\mathrm{RD}}(\cdot;\cdot)$ are not restricted to a specific form. 
EF-LIC realizes them as affine projections for efficiency.
The reconstructed image $\hat{\bm{x}}$ is decoded from the quantized latent $\hat{\bm{y}}$ as $\hat{\bm{x}}=g_s(\hat{\bm{y}})$, where $\hat{\bm{y}} = (\hat{\bm{y}}_1,\ldots,\hat{\bm{y}}_N)$.

In practice, EF-LIC adopts the modules in~\cite{dcvc-rt} for $g_a(\cdot)$, $g_s(\cdot)$, $h_a(\cdot)$, and $h_s(\cdot)$, uses its context model to realize $f_i^{\mathrm{RD}}(\cdot)$, and partitions the latent $\bm{y}$ into four quadtree-based groups $(\bm{y}_1,\bm{y}_2,\bm{y}_3,\bm{y}_4)$.
To support multiple target rates, EF-LIC realizes $Q_{\bm{z}}(\cdot)$ and $\{Q_i^{\mathrm{RD}}(\cdot)\}_{i=1}^4$ as residual vector quantizers (RVQ)~\cite{speech-dac}.
We group these RVQ-based quantizers into a quantizer set $\mathcal{Q}$, where every RVQ employs the same number of codebooks $m$.
The bitrate in bits per pixel (BPP) is
\begin{equation}
  \mathrm{BPP}=\frac{m}{f_y^2}\left(\frac{f_y^2}{f_z^2}\log K_{\bm{z}}+\frac{1}{N}\sum_{i=1}^N \log K_i\right).
\end{equation}
Here, $K_{\bm{z}}$ and $K_i$ denote the numbers of codewords per codebook in $Q_{\bm{z}}(\cdot)$ and $Q_i^{\mathrm{RD}}(\cdot)$, respectively, and $\log$ denotes $\log_2$ throughout the paper.
Moreover, we define a discrete set of RVQ codebook counts $\mathcal{M}=\{m_1,\dots,m_M\}$.
For each $m\in\mathcal{M}$, we construct a corresponding quantizer set $\mathcal{Q}^{(m)}$ in which every RVQ uses $m$ codebooks, and we select $\mathcal{Q}^{(m)}$ at inference time to obtain the desired BPP.
We provide the detailed implementation of EF-LIC and its bitstream packing method in~\cref{app:architecture}.

Since all components of EF-LIC are parallelizable, the entire codec can run efficiently on GPUs.
Next, we present a theoretical analysis showing that EF-LIC achieves compression performance comparable to LIC with entropy coding.

\subsection{Maximum-Entropy Probabilistic Shaping}

In this subsection, we analyze the statistical redundancy of the indices produced by VQ in EF-LIC.
Following information theory \cite{shannon1948mathematical}, we measure statistical redundancy using the entropy
$H(X) \triangleq -\sum_x P_X(x)\log P_X(x)$.
Here $X \sim P_X$ is a discrete random variable.
For any lossless representation of $X$, the expected encoded bitstream length $R$ satisfies $R \ge H(X)$.
Entropy coding exploits a non-uniform $P_X$ to approach this bound.

In EF-LIC, VQ indices are transmitted with fixed-length symbols, so efficiency is governed by how closely the index sequence approaches the maximum-entropy limit.
We define $J \in \{1,\ldots,K\}^n$ as the index sequence after VQ, where $n$ is the sequence length and $K$ is the codebook size.
Since $J$ is a length-$n$ sequence with $K$ possible values at each position, it can take at most $K^n$ distinct outcomes. Therefore, $H(J)\le \log(K^n)=n\log K$, with equality when $J$ is uniform over $\{1,\ldots,K\}^n$.
Under fixed-length coding, this yields an available budget of $n\log K$ bits to represent $J$.
We define the normalized entropy gap as
\begin{equation}
\label{eq:delta_H}
\Delta H \triangleq \frac{n\log K - H(J)}{n\log K}.
\end{equation}
This ratio quantifies the fraction of the fixed-length budget that is statistically redundant.
In particular, $\Delta H=0$ holds exactly when $H(J)=n\log K$, meaning that the indices achieve the maximum entropy bound.

Prior empirical studies of VQ-based codecs report that $\Delta H$ tends to decrease toward zero as training converges \cite{image-rvq,speech-dac}.
A common design across these systems is unconstrained VQ combined with end-to-end optimization for reconstruction quality.
Motivated by these findings, we provide a theoretical explanation for why such unconstrained VQ drives the index entropy toward its maximum under fixed-length coding.

\begin{proposition}[Maximum-Entropy Probabilistic Shaping]
\label{thm:distortion_entropy}
For a codec employing an unconstrained VQ with $K$ codewords and target rate $R = \log K$, any distortion-optimal quantizer $Q^*$ must satisfy the entropy constraint:
\begin{equation}
  Q^* \in \arg\min_{Q:\,H(J)\le R} \mathbb{E}\left[d(X,\hat X)\right]
  \;\Longrightarrow\;
  \Delta H = 0,
\end{equation}
where $J^*$ denotes the latent index produced by $Q^*$.
\end{proposition}

Here, $X$ is the original image and $\hat{X}$ is its reconstruction. $d(X,\hat X)$ denotes a nonnegative distortion measure. $\mathbb{E}[\cdot]$ denotes expectation.
A proof by contradiction is provided in \cref{app:proof1}.
\Cref{thm:distortion_entropy} indicates that VQ can be viewed as maximum-entropy probabilistic shaping, which pushes the induced index distribution toward uniformity and leaves little statistical redundancy.

In practice, the distortion optimality condition in \Cref{thm:distortion_entropy} can be restrictive.
A weaker but more general characterization is given in~\cite{vq-p}: for a high rate $C$-dimensional VQ optimized only for quantization error, the induced index probabilities satisfy
\begin{equation}
\label{eq:pdf_vq}
p_J(j)\ \propto\ p_Y(\bm{c}_j)^{\frac{2}{C+2}},
\end{equation}
where $p_Y$ is the probability density of $Y$ and $\bm{c}_j$ denotes the codeword indexed by $j$.
If $Y$ follows a Gaussian distribution and $C=8$, \cref{eq:pdf_vq} already yields $\Delta H \le 5\%$, which is consistent with empirical results on VQ-based codecs~\cite{vq-vae,image-rvq,speech-dac}.
If $\Delta H$ remains above $5\%$, it is preferable to redesign or retrain the quantizer rather than rely on entropy coding.

Motivated by~\cref{thm:distortion_entropy,eq:pdf_vq}, we do not impose an explicit rate constraint during training.
Instead, we regularize the quantizer using a codebook loss $\mathcal{L}_{\mathrm{cb}}$ to control the quantization error.
We train a single model across the operating points indexed by $m\in\mathcal{M}$.
\begin{equation}
  \label{eq:loss}
  \begin{aligned}
    \mathcal{L} = \frac{1}{|\mathcal{M}|}\sum_{m\in \mathcal{M}} \Big(
      \lVert \bm{x} - \hat{\bm{x}}_m\rVert_1
      + \lambda_{\mathrm{per}} \, \mathcal{L}_{\mathrm{per}}(\bm{x}, \hat{\bm{x}}_m) \\
      + \lambda_{\mathrm{adv}} \, \mathcal{L}_{\mathrm{adv}}(\bm{x}, \hat{\bm{x}}_m)
      + \lambda_{\mathrm{cb}} \, \mathcal{L}_{\mathrm{cb}}^{m}
    \Big).
  \end{aligned}
\end{equation}
Here $\hat{\bm{x}}_m$ denotes the reconstruction obtained when each quantizer $Q$ uses $m$ codebooks.
We instantiate $\mathcal{L}_{\mathrm{per}}$ as LPIPS computed with a VGG network~\cite{vgg}, and set $\mathcal{L}_{\mathrm{adv}}$ to the adaptive PatchGAN adversarial loss~\cite{vq-gan}.
Following~\cite{vq-vae}, the codebook loss $\mathcal{L}_{\mathrm{cb}}^{m}$ includes a commitment term and a codebook update term, which constrains $\Delta H$ to remain small and thereby removes statistical redundancy.

\subsection{Representation-domain Latent Decorrelation}

In this subsection, we analyze correlation redundancy in EF-LIC.
Information theory \cite{shannon1948mathematical} provides the R--D function as a principle for analyzing compression performance, which is defined as:
\begin{equation}
  \label{eq:rd_function}
D_X(R)\ \triangleq\ \inf_{P_{\hat{X}|X}:\ I(X;\hat{X})\le R}\ \mathbb{E}\!\left[d(X,\hat{X})\right].
\end{equation}
Here $D_X(R)$ denotes the minimum achievable expected distortion between $X$ and $\hat{X}$ under an average bitrate constraint $R$.
The infimum is taken over all conditional distributions $P_{\hat{X}|X}$ that satisfy $I(X;\hat{X})\le R$, where $I(X;\hat{X})$ is the mutual information between $X$ and $\hat{X}$.
This constraint limits how much information about $X$ can be preserved in $\hat{X}$, serving as a lower-bound for bitrate in theory.
Accordingly, a more effective LIC model attains a lower distortion at a given bitrate by reducing redundant information.

We first establish a baseline that excludes representation-domain decorrelation and directly quantizes each latent group independently, so as to isolate and evaluate the contribution of decorrelation in EF-LIC.

\begin{definition}[Independent Quantization (IQ)]\label{def:scheme1}
As shown in \cref{fig:framework-a}, let $\mathcal{Y}=\{Y_i\}_{i=1}^N$ denote the random variables for the latent groups $\{\bm{y}_i\}_{i=1}^N$.
The baseline VQ quantizes each $Y_i$ independently with a quantizer $Q_i^{\mathrm{IQ}}(\cdot)$ under a fixed rate $R=n\log K$.
Its R--D function $D_X^{\mathrm{IQ}}(\cdot)$ is
\begin{equation}\label{eq:iq-opt}
\begin{aligned}
D_X^{\mathrm{IQ}}(R)
&\triangleq
\inf_{\{Q_i^{\mathrm{IQ}}\}}
\ \mathbb{E}\!\left[d\!\left(X,\hat X\right)\right]\\
\text{s.t.}\quad
 \hat Y_i &= Q_i^{\mathrm{IQ}}(Y_i), \quad i=1,\dots,N, \\
 R &= n\log K. 
\end{aligned}
\end{equation}
\end{definition}\vspace{-0.3\baselineskip}

Let $D_X^{\mathrm{RD}}(R)$ denote the R--D function of EF-LIC. We compare it against $D_X^{\mathrm{IQ}}(R)$ in the following proposition.
\begin{proposition}[R--D Lower bound for EF-LIC]
\label{thm:ar_dominance_context}
Assume $e_i^{\mathrm{RD}}, d_i^{\mathrm{RD}}, Q_i^{\mathrm{RD}}$ are sufficiently expressive,
for any grouped latent $Y=(Y_1,\dots,Y_N)$, there exist
$e_i^{\mathrm{RD}}, d_i^{\mathrm{RD}}, Q_i^{\mathrm{RD}}$, $i \in \{1,\dots,N\}$, such that
\begin{equation}
  \forall R\ge 0,\quad D_X^{\mathrm{RD}}(R) \le D_X^{\mathrm{IQ}}(R).
\end{equation}
If there exists $i$ such that $I(\hat Y_{i};\hat Y_{<i})>0$, then
\begin{equation}
  \exists R\ge 0,\quad D_X^{\mathrm{RD}}(R) < D_X^{\mathrm{IQ}}(R).
\end{equation}
\end{proposition}

A proof is given in \cref{app:proof2}.
\Cref{thm:ar_dominance_context} establishes that adding representation-domain latent decorrelation cannot worsen the R--D trade-off.
Since Independent Quantization in \cref{def:scheme1} underpins several strong VQ-based codecs~\cite{vq-mao2024extreme,vq-vae,speech-soundstream,speech-dac}, EF-LIC is guaranteed to match or improve upon this baseline in terms of compression performance.

Next, we establish an upper bound for EF-LIC by comparing it with typical entropy-coded LIC, in which context modeling and entropy coding can eliminate both statistical and correlation redundancy in principle, to quantify how efficiently EF-LIC narrows the gap to entropy-coded LIC.

\begin{definition}[Probability-Domain context modeling (CM)]\label{def:scheme2}
As shown in \cref{fig:framework-b}, let $f_i^{\mathrm{CM}}$ denote the context model and $\theta_i = (\bm{\mu}_i,\bm{\sigma}_i)$ is the distribution parameter for $Y_i$.
Following~\cite{lic-joint}, the R--D function of typical LIC with entropy coding is defined as
\begin{equation}\label{eq:em-opt}
\begin{aligned}
D_X^{\mathrm{CM}}(R)
&\triangleq
\inf_{\{Q_i^{\mathrm{CM}}, f_i^{\mathrm{CM}}\}}
\ \mathbb{E}\!\left[d\!\left(X,\hat X\right)\right] \\
\text{s.t.}\quad
\hat Y_i & = Q_i^{\mathrm{CM}}(Y_i), \quad i=1,\dots,N, \\
\theta_i & = f_i^{\mathrm{CM}}(\hat Y_{<i}), \quad i=1,\dots,N, \\
R&  = \sum_{i=1}^N
  \mathbb{E}\!\left[
    -\log P_{\hat Y_i\mid \hat Y_{<i}}(\hat Y_i \mid \hat Y_{<i};\theta_i)
  \right].
\end{aligned}
\end{equation}
\end{definition}
The rate $R$ is achieved through ideal entropy coding. $Q_i^{\mathrm{CM}}(\cdot)$ is usually the round operation. We compare $D_X^{\mathrm{RD}}(R)$ against $D_X^{\mathrm{CM}}(R)$ in the following theorem.

\begin{theorem}[R--D upper bound for EF-LIC]
\label{thm:em}
Assume $e_i^{\mathrm{RD}}, d_i^{\mathrm{RD}}, Q_i^{\mathrm{RD}}$ are sufficiently expressive (i.e., $K$ is sufficiently large).
Fix a target rate $R>0$ and an arbitrary parameter $\varepsilon\in(0,1)$.
Then there exists an implementation with fixed-length rate budget $R' \triangleq \frac{R}{1-\varepsilon}$,
whose induced index distribution satisfies the normalized entropy gap bound
\begin{equation}
\Delta \bar{H} \triangleq
\frac{\sum_{i=1}^N \left( n_i \log K_i - H\left(J_i^{\mathrm{RD}}\mid \hat{Y}_{<i}^{\mathrm{RD}}\right)\right)}{\sum_{i=1}^N n_i \log K_i}
\le \varepsilon,
\label{eq:entropy_gap}
\end{equation}
and whose R--D performance obeys,
\begin{equation}
D_X^{\mathrm{RD}}\!\left(\frac{R}{1-\varepsilon}\right)\le D_X^{\mathrm{CM}}(R).
\end{equation}
According to \cref{thm:distortion_entropy}, the overhead factor $1/(1-\varepsilon)$ will be closed to $1$ with sufficiently large $K$ under sufficient training.
\end{theorem}

A proof is given in \cref{app:proof3}.
\Cref{thm:em} shows that EF-LIC removes correlation redundancy as effectively as typical LIC with context modeling and entropy coding.
Together with our analysis of statistical redundancy, this establishes that EF-LIC removes both types of redundancy while preserving compression performance.
With the architecture in \cref{sec:architecture}, EF-LIC further enables high GPU parallelism and low latency, mitigating the performance--efficiency bottleneck of entropy coding in conventional LIC.

\begin{figure*}[t]
  \centering
  \includegraphics[width=0.995\linewidth]{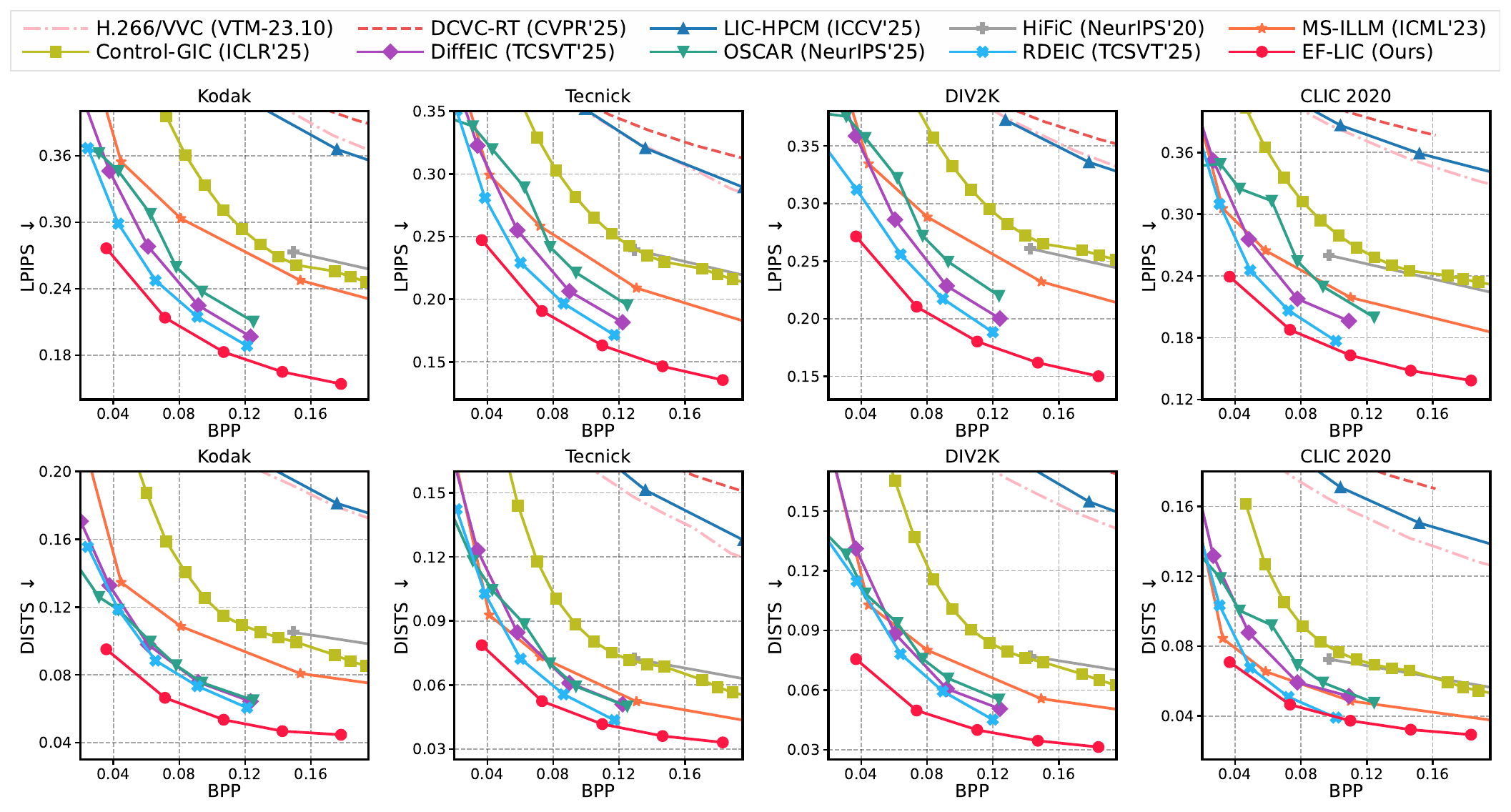}
  \caption{R--D performance on the Kodak, Tecnick, DIV2K, and CLIC2020 datasets, evaluated with LPIPS and DISTS vs. BPP. Curves closer to the origin indicate better compression performance.}
  \label{fig:RD-curve}
\end{figure*}

\begin{figure*}[t]
  \centering
  \includegraphics[width=0.995\linewidth]{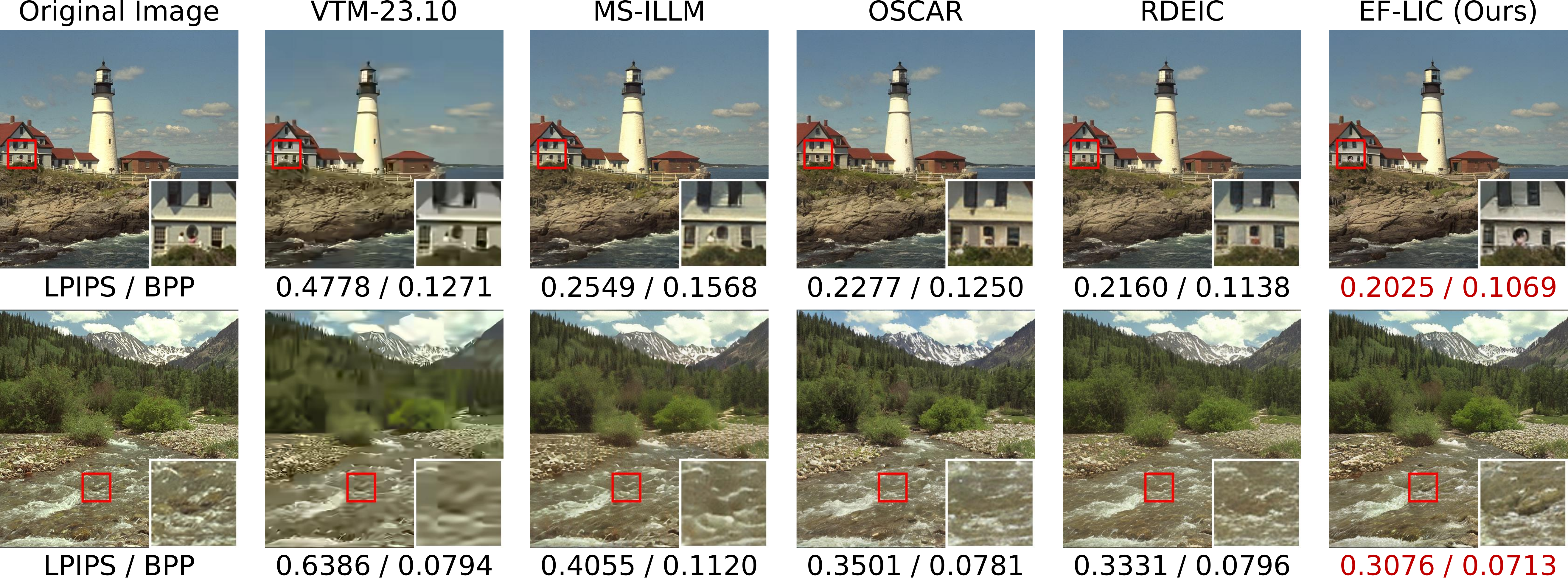}
  \caption{Visual comparison on Kodak. Numbers are LPIPS/BPP. Lower LPIPS is better.}
  \label{fig:vis}
\end{figure*}

\begin{table*}[t]
  \caption{Computational complexity measured on Kodak and BD-rate on four datasets. More negative BD-rate means lower bitrate at the same distortion. Best results are in \textbf{bold} and second best are \underline{underlined}. Dashes (--) denote unavailable results. “Enc./Dec.” reports per-image encoding/decoding time.}
  \label{tab:bd_table}
  \centering
  \begin{tabular}{@{}lrrrrrrrr@{}}
    \toprule
    \multirow{2}{*}{Method} &
    \multirow{2}{*}{Enc.(ms)} &
    \multirow{2}{*}{Dec.(ms)} &
    \multirow{2}{*}{FLOPs(G)} &
    \multirow{2}{*}{Params(M)} &
    \multicolumn{4}{c}{BD-rate (LPIPS)} \\
    \cmidrule(lr){6-9}
    & & & & & Kodak & Tecnick & DIV2K & CLIC 2020 \\
    \midrule
    H.266/VVC& $>$9999       & 150.30       & --       & --  & 313.84\%  & 295.19\%  & 285.10\%  & 498.64\%  \\
    LIC-HPCM        & 62.37   & 82.88   & 732.47  & 68.50   & 274.50\%  & 305.42\%  & 267.18\%  & 745.04\%  \\
    DCVC-RT & \underline{14.09}   & 17.08   & 382.98 & 45.65    & 393.72\%  & 329.44\%  & 349.39\%  & 584.41\%  \\
    \midrule
    HiFiC           & 526.51  & 1408.60 & 599.51  & 181.60  & 45.82\% & 68.66\% & 46.36\% & 86.45\% \\
    Control-GIC     & 103.56  & 436.26  & 5816.37 & 130.36  & 33.36\% & 68.83\% & 73.77\% & 136.25\% \\
    MS-ILLM         & 165.38  & 147.79  & 599.52  & 181.40  & 0.00\%    & 0.00\%    & 0.00\%    & 0.00\%    \\
    DiffEIC        & 210.18 & 4661.74 & 57339.93  & 1379.50      & -37.71\% & -9.96\% & -15.76\% & 4.34\% \\
    OSCAR           & 53.04   & 167.56  & 6485.61 & 1009.30 & -37.31\% & -5.76\% & -14.51\% & 18.76\% \\
    RDEIC          & 157.25  & 426.68  & 7767.46 & 1380.27  & -52.08\% & -31.36\% & -35.70\% & -20.44\% \\
    \midrule
    EF-LIC-s  & \textbf{9.94}    & \textbf{6.26}    & \textbf{179.19}  & \textbf{11.51}   & \underline{-55.38\%} & \underline{-38.23\%} & \underline{-47.36\%} & \underline{-42.10\%} \\
    EF-LIC   & 17.62    & \underline{13.72}   & \underline{279.61}  & \underline{35.74}   & \textbf{-67.86\%} & \textbf{-55.46\%} & \textbf{-62.33\%} & \textbf{-63.22\%} \\
    \bottomrule
  \end{tabular}
\end{table*}

\section{Experiments}
\subsection{Experimental Setup}
We follow the common practice~\cite{lic-balle2018variational,dcvc-rt} and set $f_y=16$ and $f_z=64$. 
Since $N=4$, we set $K_1 = 1024, K_2 = 512, K_3 = 256, K_4 = 128, K_{\bm{z}} = 1024$. 
This is an empirical setting, for which we conduct an ablation study in~\cref{tab:Ablation_codewords}. 
We also build a small model EF-LIC-s, for which we discard the hyperprior and set $K_1 = 1024, K_2 = 256, K_3 = 128, K_4 = 64$ to speed up. 
$g_a$ and $g_s$ are also simplified in it. 
We set $\mathcal{M} = \{1, 2, 3, 4, 5\}$ to cover a feasible rate range, which supports the comparison with other LIC optimized for visual quality.

We perform training on the ImageNet dataset~\cite{imagenet}. For data preprocessing, we randomly sample 1\% of the instances per epoch and apply augmentations including $256 \times 256$ random cropping and horizontal flipping. The model is optimized using Adam~\cite{adam} with $\beta_1 = 0.5$ and $\beta_2 = 0.9$. We employ a batch size of 16 for a total of 2M iterations. The learning rate is initialized at $10^{-4}$ and decayed to $10^{-5}$ after 1.5M steps. All training is conducted on one NVIDIA A100 GPU, with a peak memory footprint of approximately 10.5GB.

Evaluations are conducted on four standard test sets: (i) Kodak~\cite{Kodak} (24 images, with resolution of $768 \times 512$), (ii) Tecnick~\cite{tecnick} (100 images, with resolution of $1200 \times 1200$), (iii) DIV2K~\cite{div2k} (100 images, 2K resolution), and (iv) CLIC 2020 Professional~\cite{clic2020} (250 images, variable resolutions up to 2K).
Consistent with~\cite{vq-gic}, we report LPIPS~\cite{lpips} and DISTS~\cite{dists} as principal metrics, since they better reflect visual quality than pixel-wise metrics such as PSNR~\cite{perception}.
Therefore, we primarily compare against LIC optimized for visual quality for fairness.
We provide more results on other metrics in~\cref{app:others}.

\subsection{Rate-Distortion Performance}
The comparison includes: (i) traditional codecs: VTM-23.10~\cite{vtm}. (ii) LICs for pixel-level reconstruction: LIC-HPCM~\cite{hpcm} and DCVC-RT~\cite{dcvc-rt}. (iii) Generative LICs, including GAN-based methods: HiFiC~\cite{hific} and MS-ILLM~\cite{ms-illm}; VQ-based method: Control-GIC~\cite{cgic}; and diffusion-based methods: DiffEIC~\cite{diffeic}, OSCAR~\cite{oscar} and RDEIC~\cite{rdeic}. For VTM-23.10~\cite{vtm} and DCVC-RT~\cite{dcvc-rt}, we use their intra-frame coding schemes for image compression. To ensure a rigorous comparison, all evaluations utilize official pre-trained checkpoints in FP32 precision with a batch size of 1. Experiments are conducted on a unified hardware platform with one NVIDIA A100 GPU and an AMD EPYC 7763 CPU. Notably, under the official inference setting, evaluating OSCAR on DIV2K and CLIC 2020 requires more than 80GB of GPU memory per image.
We therefore offload selected model components to CPU memory during inference to avoid out-of-memory failures.
The results are summarized in \cref{tab:bd_table} and \cref{fig:RD-curve}, and additional evaluations with BD-rate on DISTS are detailed in~\cref{app:dists}.

Notably, EF-LIC improves BD-rate exceeding 55\% in LPIPS compared to MS-ILLM consistently across all benchmarks.
It also outperforms diffusion-based methods such as OSCAR and RDEIC, while requiring significantly fewer parameters.
Visual comparisons in \cref{fig:vis} illustrate that EF-LIC uniquely preserves the circular archway in the first image, and the authentic wave texture in the second.

\subsection{Complexity Analysis}
As shown in \cref{tab:bd_table}, we report coding time (ms), floating-point operations (GFLOPs), and model size in parameters (M), all measured on the standardized hardware described above.
Results at higher resolutions (1080p, 2K, and 4K) are reported in \cref{app:runtime}.

EF-LIC provides over $9\times$ faster encoding and $10\times$ faster decoding than MS-ILLM.
It outperforms the one-step diffusion method OSCAR while achieving $10\times$ faster decoding.
The results indicate that EF-LIC and EF-LIC-s improve compression performance while delivering an order-of-magnitude speedup over prior methods.

\subsection{Ablation Studies}\label{sec:ablation}
We next conduct ablation studies to isolate the contribution of each component.
For efficiency, all ablation models are trained on ImageNet for 1M iterations with a batch size of 16, while keeping all other training settings the same as in the main experiments.
We evaluate all variants on Kodak using LPIPS for a unified comparison.

\paragraph{Comparison with Different Variants.}

To set up, we follow the rANS~\cite{rans} in CompressAI~\cite{compressai} to implement entropy coding. To isolate module-specific impacts given varying multi-rate implementations, all models are trained for several single rates using the same loss. More detailed configurations are in~\cref{app:exp_detail}.

\begin{table}[t]
  \caption{Ablation study of EF-LIC and its variants. ``VQ'' is the baseline without inter-latent decorrelation. $\Delta$FLOPs is the FLOPs change compared to the VQ baseline. ``EC'' denotes entropy coding. ``UQ+EC'' corresponds to typical LIC with entropy coding.}
  \label{tab:Ablation_module}
  \centering
  \begin{tabular}{@{}lrrrr@{}}
    \toprule
    Modules & BD-rate & $\Delta$FLOPs & Enc.(ms) & Dec.(ms)\\
    \midrule
    VQ                      & 0.00\%       &  0.00\%   &  5.51 & 7.06    \\
    \midrule
    VQ+EC                & -14.73\%     & +4.30\%  &  362.07 & 300.83  \\
    UQ+EC                & -20.73\%     & +7.53\%  & 63.12 & 71.72   \\ 
    EF-LIC                  & \textbf{-22.20\%}   & +7.54\%  & 17.62 & 13.72   \\
    \midrule
    EF-LIC-s                & -10.76\%     &-56.30\%   & 9.94 & \textbf{6.26}    \\
    \bottomrule
  \end{tabular}
\end{table}

\begin{table}[t]
  \caption{Ablation study of per-module running time (ms). ``Q'' is quantization. ``Others'' include all remaining modules such as $g_a$ and $g_s$. ``Autoregressive'' is the context-conditional transform in EF-LIC or the context model in typical LIC with entropy coding.}
  \label{tab:Ablation_time}
  \centering
  \begingroup
  \setlength{\tabcolsep}{3.7pt} 
  \begin{tabular}{@{}lrrrrr@{}}
    \toprule
      & Q & EC & Autoregressive & Others & Total \\
    \midrule
    VQ       & 1.40  & --    & --    & 9.48 & 12.52   \\
    \midrule
    UQ+EC  & 0.01 & 108.60&  4.42 & 10.73 & 124.74 \\
    VQ+EC  & 1.52 & 507.89&  3.29 & 11.09 & 525.09 \\
    EF-LIC    & 9.16 & --    &  4.38 & 10.53 & 24.78  \\
    \midrule
    EF-LIC-s  & 6.93 & --    &  2.51 &  3.68 & 13.84  \\
    \bottomrule
  \end{tabular}
  \endgroup
\end{table}

We first compare EF-LIC with the VQ baseline without decorrelation, reported as ``VQ'' in \cref{tab:Ablation_module,tab:Ablation_time,fig:bd2}.
The results suggest representation-domain decorrelation improves BD-rate by $22.20\%$, suggesting that it effectively removes correlation redundancy, supporting \cref{thm:ar_dominance_context}.
The runtime breakdown in \cref{tab:Ablation_time} shows the autoregressive module contributes only a small fraction of the combined runtime, indicating its efficiency.
Because the autoregressive transform introduces additional computation,
we evaluate EF-LIC-s, a lightweight variant configured to match the decoding latency of the VQ baseline to form a fair comparison.
Under this setting, EF-LIC-s still reduces BD-rate by $10.76\%$, indicating that the gain comes from decorrelation rather than increased computation.

We next compare EF-LIC with its entropy-coded variant, reported as ``UQ+EC'' in \cref{tab:Ablation_module,tab:Ablation_time,fig:bd2}.
EF-LIC achieves better compression performance while decoding about $5\times$ faster than ``UQ+EC'' because of long entropy coding time in ``UQ+EC''.
\cref{thm:em} predicts that ``UQ+EC'' can outperform EF-LIC by at most the remaining entropy gap.
Though EF-LIC exhibits a small average gap of $\Delta \bar{H}=3.42\%$ (detailed results are in \cref{app:entropy}), the use of rANS in ``UQ+EC'' introduces extra redundancy and worsens BD-rate by $3.28\%$ compared to ideal entropy coding, which is consistent with the experimental results to \cref{thm:em}.

Finally, we apply a context model and entropy coding directly to the VQ indices \cite{vq-pq-mim}, and report the results as ``VQ+EC'' in \cref{tab:Ablation_module,tab:Ablation_time,fig:bd2}.
This approach is impractical in our setting because entropy coding must construct input-dependent cumulative distribution functions, which leads to very long coding time.
Moreover, the hard VQ operation blocks gradients to the context model, making end-to-end optimization suboptimal.

\paragraph{Ablation Study on Codeword Numbers.}

The codebook sizes $[\log K_1, \log K_2, \log K_3, \log K_4, \log K_{\bm{z}}]$ for the quantizers $[Q_1, Q_2, Q_3, Q_4, Q_{\bm{z}}]$ are manually specified.
We conduct an ablation study on these configurations. 
As reported in \cref{tab:Ablation_codewords}, we find that using the hyperprior can already provide primary performance gain with side information.

After introducing the context-conditional autoregressive transform, allocating fewer codewords to the later quantizers tends to improve performance.
This is because the later latents contain less information. Smaller codebooks better match this reduced support.
Overly small codebooks can become a bottleneck and degrade performance.

\begin{table}[t]
  \caption{Ablation study of the codebook sizes for quantizers $[Q_1, Q_2, Q_3, Q_4, Q_{\bm{z}}]$. ``Hyper'' denotes the hyperprior. The column $K$ reports the corresponding logarithmic codebook configuration $[\log K_1, \log K_2, \log K_3, \log K_4, \log K_{\bm{z}}]$.}
  \label{tab:Ablation_codewords}
  \centering
  \newcommand{\n}[1]{\makebox[1.1em][c]{#1}}
  \newcommand{\dsh}{\makebox[1.1em][c]{--}}
  \begin{tabular}{@{}lrr@{}}
    \toprule
    Modules & $K$ & BD-rate \\
    \midrule
    VQ             & $[\n{10}, \n{10},   \n{10},   \n{10},   \dsh]$   & 0.00\%     \\
    \midrule
    +Hyper               & $[\n{10}, \n{10}, \n{10}, \n{10},   \n{10}]$ & -9.37\%\\
    +Hyper+Autoregressive      & $[\n{10}, \n{10}, \n{10}, \n{10}, \n{10}]$ & -16.72\%\\ 
    +Hyper+Autoregressive      & $[\n{10}, \n{9},  \n{8},  \n{7},  \n{10}]$ & \textbf{-28.60\%}\\
    +Hyper+Autoregressive      & $[\n{10}, \n{8},  \n{6},  \n{6},  \n{10}]$ & -28.35\%\\
    \bottomrule
  \end{tabular}
\end{table}

\section{Applications}
A major practical limitation of existing LIC systems is that they require hybrid GPU--CPU execution, which prevents the model from being exported as a unified computation graph, such as ONNX, and thus complicates deployment on real devices \cite{vq-mcquic}. By eliminating entropy coding, EF-LIC removes the CPU-side dependency and enables end-to-end inference within a single accelerator-friendly computation graph, which greatly simplifies deployment. Building on this advantage, we successfully export EF-LIC as \textit{self-contained ONNX and TorchScript models}, and deploy it on embedded devices and smartphones. This level of portability is not achievable for entropy-coded LIC systems.

EF-LIC also improves numerical robustness across heterogeneous devices. Existing entropy-coded LIC systems require the encoder and decoder to produce exactly matched entropy-model probabilities. In cross-device deployment, however, tiny numerical discrepancies in floating-point computation may change the cumulative distribution functions used by entropy coding, desynchronize the bitstream, and eventually cause decoding failure. This issue has also been reported in DCVC-RT~\cite{dcvc-rt}. Since EF-LIC removes entropy coding and transmits fixed-length VQ indices, its decoding process does not depend on reproducing device-specific entropy-model probabilities. As a result, EF-LIC supports reliable cross-platform image encoding and decoding across different hardware backends.

\section{Limitations}
This work focuses on theoretically validating the effectiveness of EF-LIC under a reasonable distortion regime, and several engineering aspects remain to be improved.
First, the codebook sizes in EF-LIC are currently hand-designed.
Second, although RVQ is significantly faster than entropy coding, its runtime is still non-negligible and needs acceleration.
Third, while RVQ yields strong visual quality, its performance under pixel-wise criteria such as PSNR is less competitive.
Nevertheless, these limitations are orthogonal to the main purpose of this paper, and we leave further engineering optimizations to future work.

\section{Conclusion}
In this paper, we present EF-LIC to break the runtime bottleneck in typical LIC.
EF-LIC reduces statistical redundancy via unconstrained VQ and reduces correlation redundancy via a context-conditional autoregressive transform, while enabling flexible multi-rate operation.
We theoretically show that the resulting approach can match the compression performance of typical LIC.
Experiments demonstrate improved compression performance and substantially lower coding latency compared with state-of-the-art methods and several variants, validating EF-LIC as a new paradigm for LIC without entropy coding.

\section*{Software and Data}
The source code of EF-LIC is publicly available at \url{https://github.com/SevenCTHU/EF-LIC}.

\section*{Acknowledgments}
This work is supported by the National Key Research and Development Program of China under Grant No. 2023YFB2904300, the National Natural Science Foundation of China under Grant No. 62293484, No. 62441235, and No. 92570204, Beijing Natural Science Foundation (F251001 and L257005).

\section*{Impact Statement}
Entropy coding is ubiquitous in both traditional and learned image compression, but its sequential processing nature is difficult to parallelize on GPUs and limits throughput.
This work provides theoretical evidence that key redundancies in images can be reduced without entropy coding, and it instantiates this idea with a multi-rate entropy-coding-free codec that achieves competitive compression performance with lower coding latency.
By enabling lower-latency and more compute-efficient compression, this work may benefit real-time and on-device imaging applications.
To summarize, our contribution lies in establishing a theoretical and practical foundation for efficient learned image compression without entropy coding, paving the way for low-latency image compression.

\bibliography{main}
\bibliographystyle{icml2026}

\clearpage
\appendix
\section*{Appendix}
In the appendix, we provide the following:
\begin{itemize}
  \item \Cref{app:sec1} provides proofs of \cref{thm:distortion_entropy,thm:ar_dominance_context,thm:em}.
  \item \Cref{app:architecture} describes the model implementation and bitstream packing methods.
  \item \Cref{app:sec3} presents additional experimental details, including the exact settings of competing methods and the training losses used in our ablations.
  \item \Cref{app:sec4} reports additional results, including further entropy-gap analysis, BD-rate results on DISTS, results under more metrics, more runtime tests, and an additional LPIPS-based comparison with recent generative codecs.
\end{itemize}

\section{Proof of Theorems}\label{app:sec1}
In the main text, we present~\cref{thm:distortion_entropy,thm:ar_dominance_context,thm:em}, which form the theoretical basis of the proposed EF-LIC. This section provides detailed proofs.
\subsection{Proof of \cref{thm:distortion_entropy}}
\label{app:proof1}
\begin{proof}
Let $Q^*$ be any quantizer that attains the minimal distortion under the constraint $H(J)\le R$.
Recall \cref{eq:rd_function}, the R--D function of source $X$ is defined as
\begin{equation*}
D_X(R)\ =\ \inf_{P_{\hat{X}|X}:\ I(X;\hat{X})\le R}\ \mathbb{E}\!\left[d(X,\hat{X})\right].
\end{equation*}
For a well-defined distortion measure $d(X,\hat X)$, the R--D function is strictly decreasing over the distortion range of interest.
Consequently, its generalized inverse is well defined, which we denote by $R_X(D)$.
\begin{equation}
R_X(D)\ \triangleq\ \inf_{P_{\hat{X}|X}:\ \mathbb{E}\left[d(X,\hat{X})\right]\le D}\ I(X;\hat{X}).
\end{equation}
This means
\begin{equation}
  I(X;\hat X^*) \;\ge\; R_X(D^*).
\end{equation}
Since $J^*$ is a function of $X$ and $\hat X^*$ is a function of $J^*$, $X \to J^* \to \hat X^*$ forms a Markov chain and hence
\begin{equation}
  I(X;\hat X^*) \;\le\; I(X;J^*) \;\le\; H(J^*).
\end{equation}
Combining the two inequalities yields
\begin{equation}
  H(J^*) \;\ge\; R_X(D^*).
\end{equation}
On the other hand, from the definition of $D_X(R)$ we have
\begin{equation}
  D^* = D_X(R) \quad\Longrightarrow\quad R_X(D^*) \;\le\; R.
\end{equation}
Thus
\begin{equation}
  R_X(D^*) \;\le\; H(J^*) \;\le\; R.
  \label{eq:entropy-bracket}
\end{equation}

Since $R_X(D)$ is strictly decreasing on the distortion range of interest, its generalized inverse $D_X(R)$ is strictly decreasing in $R$. Hence, for any $R' < R$,
\begin{equation}
  D_X(R') > D_X(R) = D^*.
\end{equation}

Suppose, for the sake of contradiction, that $H(J^*) < R$. Choose any $R'$ such that
\begin{equation}
  H(J^*) \;\le\; R' \;<\; R.
\end{equation}
Because $Q^*$ satisfies $H(J^*) \le R'$, it is feasible for the optimization defining $D_X(R')$, so
\begin{equation}
  D_X(R') \;\le\; D^*.
\end{equation}
Combining this with the strict monotonicity of $D_X(\cdot)$, we obtain
\begin{equation}
  D_X(R') > D_X(R) = D^*,
\end{equation}
a contradiction. Therefore $H(J^*)$ cannot be strictly smaller than $R$, and together with $H(J^*) \le R$ this implies
\begin{equation}
  H(J^*) = R = n\log K.
\end{equation}
Using the definition of $\Delta H$ in \cref{eq:delta_H}, we have $\Delta H = 0$,
which completes the proof.
\end{proof}

\subsection{Proof of \cref{thm:ar_dominance_context}}
\label{app:proof2}

\begin{proof}

Obviously, when $e_i^{\mathrm{RD}}$ and $d_i^{\mathrm{RD}}$ are chosen as identity mappings and
$Q_i^{\mathrm{RD}}=Q_i^{\mathrm{IQ}}$ for all $i$, Scheme RD reduces to Scheme IQ.
Hence, for any rate $R$, every reconstruction achievable by IQ is also achievable by RD.
Therefore, the feasible set of RD contains that of IQ, which implies
\begin{equation}
D_X^{\mathrm{RD}}(R) \le D_X^{\mathrm{IQ}}(R), \quad \forall R \ge 0 .
\end{equation}

Next, at rate $R$ there exist a distortion level $D^\star$ on the distortion range of interest and an
IQ scheme achieving $R_X^{\mathrm{IQ}}(D^\star)$ such that the induced reconstruction
$\hat Y=(\hat Y_1,\dots,\hat Y_N)$ satisfies, for some $i$,
\begin{equation}
I(\hat Y_i;\hat Y_{<i})>0 .
\label{eq:mi_assump_new}
\end{equation}
Here we denote the generalized inverse
\begin{equation}
\begin{aligned}
R_X^{\mathrm{IQ}}(D)\triangleq \inf\{R\ge 0:\ D_X^{\mathrm{IQ}}(R)\le D\},
\\
R_X^{\mathrm{RD}}(D)\triangleq \inf\{R\ge 0:\ D_X^{\mathrm{RD}}(R)\le D\}.
\end{aligned}
\end{equation}
Let $S\triangleq \hat Y_{<i}$, so \eqref{eq:mi_assump_new} gives $I(\hat Y_i;S)>0$.

Since we are under Scheme IQ, the $i$-th group does not use $S$ when producing $\hat Y_i$.
Equivalently, $\hat Y_i \!\perp\!\!\!\perp S \mid Y_i$, and thus $S \to Y_i \to \hat Y_i$ forms a Markov chain.
By the data processing inequality,
\begin{equation}
I(\hat Y_i;S) \le I(Y_i;S).
\end{equation}
Therefore $I(\hat Y_i;S)>0$ implies $I(Y_i;S)>0$, meaning the side information is non-trivial.

Fix the coding rules of all groups $j\neq i$ in the above IQ scheme, and denote the resulting $\hat Y_{>i}$.
Define the induced side-information-dependent distortion
\begin{equation}
\begin{aligned}
&\bar d_i(y_i,\hat y_i,s) \triangleq \\
&\mathbb{E}\!\left[d\!\left(X,\ g_s(s,\hat y_i,\hat Y_{>i})\right)\ \middle|\ Y_i=y_i,\ S=s\right].
\end{aligned}
\end{equation}
Then for any choice of the $i$-th group, the overall distortion equals
$\mathbb{E}[\bar d_i(Y_i,\hat Y_i,S)]$ under the fixed rules of other groups.

Define the conditional R--D function with two-sided side information $S$ as
\begin{equation}
R_{i\mid S}(D)
\triangleq
\inf_{P_{\hat Y_i\mid Y_i,S}:\, \mathbb{E}[\bar d_i(Y_i,\hat Y_i,S)]\le D}
I(Y_i;\hat Y_i\mid S),
\end{equation}
and the counterpart without using $S$ as
\begin{equation}
R_i(D)
\triangleq
\inf_{P_{\hat Y_i\mid Y_i}:\, \mathbb{E}[\bar d_i(Y_i,\hat Y_i,S)]\le D}
I(Y_i;\hat Y_i).
\end{equation}
Since any $P_{\hat Y_i\mid Y_i}$ can be embedded into the conditional class by ignoring $S$,
\begin{equation}
R_{i\mid S}(D)\le R_i(D),\quad \forall D .
\label{eq:rd_leq}
\end{equation}
Moreover, $I(Y_i;S)>0$ shows that the side information is non-trivial.
Under the standard strictness result for two-sided side information with side-information-dependent distortion
\cite{linder2002source}, there exists (and we fix) the above $D^\star$ such that
\begin{equation}
R_{i\mid S}(D^\star) < R_i(D^\star).
\label{eq:rd_strict}
\end{equation}
Let $\delta \triangleq R_i(D^\star)-R_{i\mid S}(D^\star)>0$.

By the operational fixed-length rate--distortion theorem \cite{shannon1959coding},
for any $\epsilon>0$, any fixed-length code that does not use $S$ and achieves distortion at most $D^\star$
must have rate at least $R_i(D^\star)-\epsilon$, while there exists a fixed-length code using $S$ at both
encoder and decoder achieving distortion at most $D^\star$ with rate at most $R_{i\mid S}(D^\star)+\epsilon$.
Replacing only the $i$-th group in the above IQ scheme by such a two-sided side-information code (implemented
by sufficiently expressive $e_i^{\mathrm{RD}},d_i^{\mathrm{RD}},Q_i^{\mathrm{RD}}$) and keeping all
other groups unchanged yields an RD scheme achieving distortion at most $D^\star$ with total rate at most
$R_X^{\mathrm{IQ}}(D^\star)-\delta+2\epsilon$.
Letting $\epsilon\downarrow 0$, we obtain
\begin{equation}
R_X^{\mathrm{RD}}(D^\star) \le R_X^{\mathrm{IQ}}(D^\star)-\delta
< R_X^{\mathrm{IQ}}(D^\star).
\label{eq:global_rate_strict}
\end{equation}

Choose $R \triangleq R_X^{\mathrm{IQ}}(D^\star)-\delta/2$. Then $R\ge R_X^{\mathrm{RD}}(D^\star)$ by
\eqref{eq:global_rate_strict}, hence $D_X^{\mathrm{RD}}(R)\le D^\star$ by definition of generalized inverse.
On the other hand, since $R < R_X^{\mathrm{IQ}}(D^\star)$, we must have $D_X^{\mathrm{IQ}}(R)>D^\star$
(otherwise $R$ would belong to the set $\{r:\ D_X^{\mathrm{IQ}}(r)\le D^\star\}$ contradicting the definition
of $R_X^{\mathrm{IQ}}(D^\star)$). Therefore,
\begin{equation}
D_X^{\mathrm{RD}}(R) < D_X^{\mathrm{IQ}}(R)
\end{equation}
for some $R>0$, completing the proof.

\end{proof}

\subsection{Proof of \cref{thm:em}}
\label{app:proof3}
To prove \cref{thm:em}, we first give the following lemma.
\begin{lemma}
\label{lem:cpit}
Assume that $V$ takes values in a standard Borel space and that the conditional law
$V \mid (Z=z)$ is atomless for $P_Z$-almost every $z$.
Then there exists a measurable map $\Psi$ such that
\[
U \triangleq \Psi(V,Z)
\]
satisfies $U \mid (Z=z) \sim \mathrm{Unif}[0,1]$ for $P_Z$-almost every $z$.

As a consequence, for any integer $M \ge 1$, the random variable
\[
B \triangleq 1+\lfloor MU\rfloor \in \{1,\dots,M\}
\]
is conditionally uniform on $\{1,\dots,M\}$ given $Z$.
Moreover, $B$ is a measurable function of $(V,Z)$.
\end{lemma}
This statement is a standard conditional version of the probability integral transform (CPIT)
and is closely related to Rosenblatt's transform~\cite{Rosenblatt1952}.

Then we give the main proof of \cref{thm:em}.

\begin{proof}
Fix $R>0$ and an arbitrary $\varepsilon\in(0,1)$.
Let $\{g_a,g_s,\{Q_i^{\mathrm{CM}},f_i^{\mathrm{CM}}\}_{i=1}^N\}$ be a CM scheme feasible at rate $R$ that attains $D_X^{\mathrm{CM}}(R)$.
For each $i\in\{1,\dots,N\}$, define the CM symbol and context by
\begin{equation}
S_i \triangleq \hat Y_i^{\mathrm{CM}}=Q_i^{\mathrm{CM}}(Y_i),
\end{equation}
\begin{equation}
Z_i \triangleq \hat Y_{<i}^{\mathrm{CM}}.
\end{equation}
The CM rate constraint in \eqref{eq:em-opt} is a conditional cross-entropy under the learned model.
Define the intrinsic conditional entropy
\begin{equation}
R_0 \triangleq \sum_{i=1}^N H(S_i\mid Z_i).
\end{equation}
Assume the learned model assigns positive probability to every symbol in the support of $S_i$ given $Z_i$.
Then cross-entropy dominates entropy and CM feasibility at rate $R$ implies
\begin{equation}
R_0 \le R.
\end{equation}

Set the RD fixed-length rate budget
\begin{equation}
R' \triangleq \frac{R}{1-\varepsilon}.
\end{equation}
Choose integers $\{K_i\}_{i=1}^N$ and define $\mathcal{J}_i \triangleq \{1,\dots,K_i\}^{n_i}$.
We impose the fixed-length equality
\begin{equation}
\sum_{i=1}^N n_i \log K_i = R'.
\end{equation}
Such a choice exists up to a rounding slack that is at most a constant number of bits.
This slack is negligible and can be made arbitrarily small by standard blocklength scaling; in particular,
it does not prevent taking $\varepsilon$ arbitrarily small.

We now allocate additional uniform indices that will fill the entropy gap without changing the reconstruction.
Choose integers $\{M_i\}_{i=1}^N$ such that
\begin{equation}
\sum_{i=1}^N \log M_i \ge R-R_0,
\end{equation}
and such that the slack in this inequality is negligible by the same rounding argument.

We assume the CM quantizers use finite alphabets, as in practical systems.
This means $\mathrm{supp}(S_i)$ is finite for each $i$.
We also assume the fixed-length budgets are compatible with these alphabets, so that after choosing $\{K_i\}$ and $\{M_i\}$ we have
\begin{equation}
|\mathrm{supp}(S_i)|\, M_i \le K_i^{n_i}
\qquad
\text{for every } i.
\end{equation}
Define
\begin{equation}
\mathcal{T}_i \triangleq \mathrm{supp}(S_i)\times\{1,\dots,M_i\}.
\end{equation}
Fix an injective map
\begin{equation}
\iota_i:\ \mathcal{T}_i\ \hookrightarrow\ \mathcal{J}_i.
\end{equation}

We now extract uniform randomness from the continuous residual variability.
Assume that the conditional law $Y_i\mid(Z_i,S_i)$ is atomless for $P$-almost every $(Z_i,S_i)$.
By \cref{lem:cpit}, there exists a measurable map $\Psi_i$ such that
\begin{equation}
U_i \triangleq \Psi_i(Y_i,Z_i,S_i)
\end{equation}
satisfies $U_i\mid(Z_i,S_i)\sim \mathrm{Unif}[0,1]$ almost surely.
Define
\begin{equation}
B_i \triangleq 1+\lfloor M_i U_i\rfloor.
\end{equation}
Then $B_i\mid(Z_i,S_i)$ is uniform on $\{1,\dots,M_i\}$.
It follows that
\begin{equation}
H(B_i\mid Z_i,S_i)=\log M_i,
\end{equation}
and therefore
\begin{equation}
H(S_i,B_i\mid Z_i)=H(S_i\mid Z_i)+\log M_i.
\end{equation}

We now construct an RD scheme at fixed-length rate $R'$ that reproduces the CM reconstruction exactly.
Fix an injective embedding $\phi_i:\mathcal{J}_i\to\mathbb{R}^{n_i}$.
Choose a deterministic quantizer $Q_i^{\mathrm{RD}}$ such that
\begin{equation}
Q_i^{\mathrm{RD}}(\phi_i(j))=j
\qquad
\text{for all } j\in\mathcal{J}_i.
\end{equation}
Define the RD encoder by
\begin{equation}
Y_i' \triangleq e_i^{\mathrm{RD}}(Y_i,\hat Y_{<i}^{\mathrm{RD}})
\triangleq
\phi_i\!\big(\iota_i(S_i,B_i)\big).
\end{equation}
Define the RD quantizer output index and codeword by
\begin{equation}
J_i^{\mathrm{RD}} \triangleq Q_i^{\mathrm{RD}}(Y_i'),
\end{equation}
and fix a bijection $C_i:\mathcal{J}_i\to\mathcal{C}_i$ and set
\begin{equation}
\hat Y_i' \triangleq Q_i^{\mathrm{RD}}(Y_i') \triangleq C_i(J_i^{\mathrm{RD}}).
\end{equation}
Define the RD decoder to recover $(S_i,B_i)$ and output only the CM symbol
\begin{equation}
\hat Y_i \triangleq d_i^{\mathrm{RD}}(\hat Y_i',\hat Y_{<i}^{\mathrm{RD}})
\triangleq
\pi_S\!\Big(\iota_i^{-1}\!\big(C_i^{-1}(\hat Y_i')\big)\Big),
\end{equation}
where $\pi_S$ projects $(S_i,B_i)$ onto $S_i$.

By construction, $\hat Y_i^{\mathrm{RD}}=\hat Y_i^{\mathrm{CM}}$ for every $i$.
Therefore
\begin{equation}
\hat Y^{\mathrm{RD}}=\hat Y^{\mathrm{CM}}.
\end{equation}
Applying the same synthesis transform yields
\begin{equation}
\hat X^{\mathrm{RD}}=g_s(\hat Y^{\mathrm{RD}})=g_s(\hat Y^{\mathrm{CM}})=\hat X^{\mathrm{CM}}.
\end{equation}
Hence the distortions coincide
\begin{equation}
\mathbb{E}\big[d(X,\hat X^{\mathrm{RD}})\big]
=
\mathbb{E}\big[d(X,\hat X^{\mathrm{CM}})\big]
=
D_X^{\mathrm{CM}}(R).
\end{equation}

\begin{figure*}[t]
  \centering
  \begin{subfigure}[c]{0.67\linewidth}
    \includegraphics[width=0.99\linewidth]{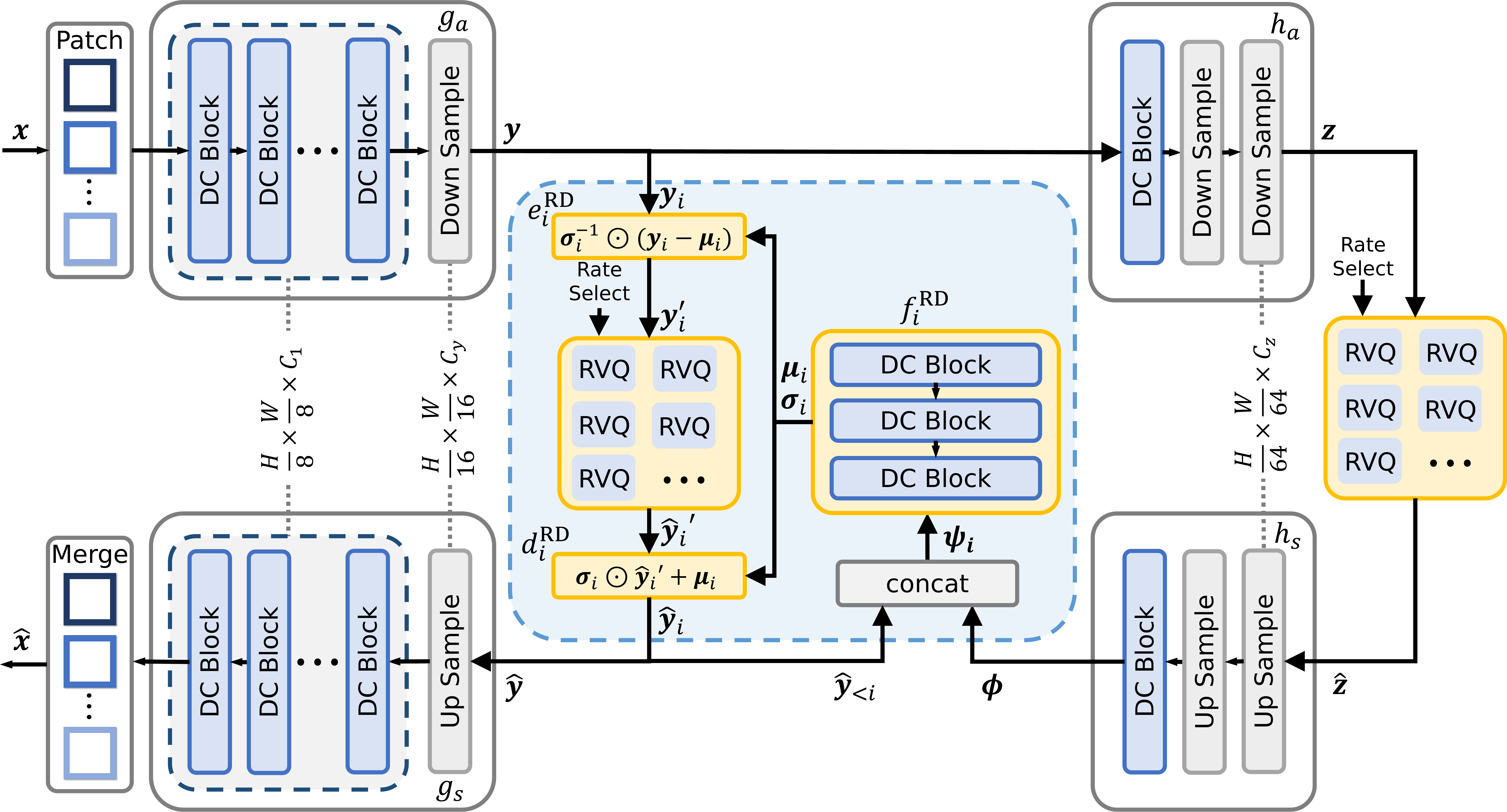}
    \caption{The detailed implementation of EF-LIC.}
    \label{fig:short-a}
  \end{subfigure}
  \hfill
  \begin{subfigure}[c]{0.315\linewidth}
    \begin{subfigure}{0.98\linewidth}
      \includegraphics[width=0.99\linewidth]{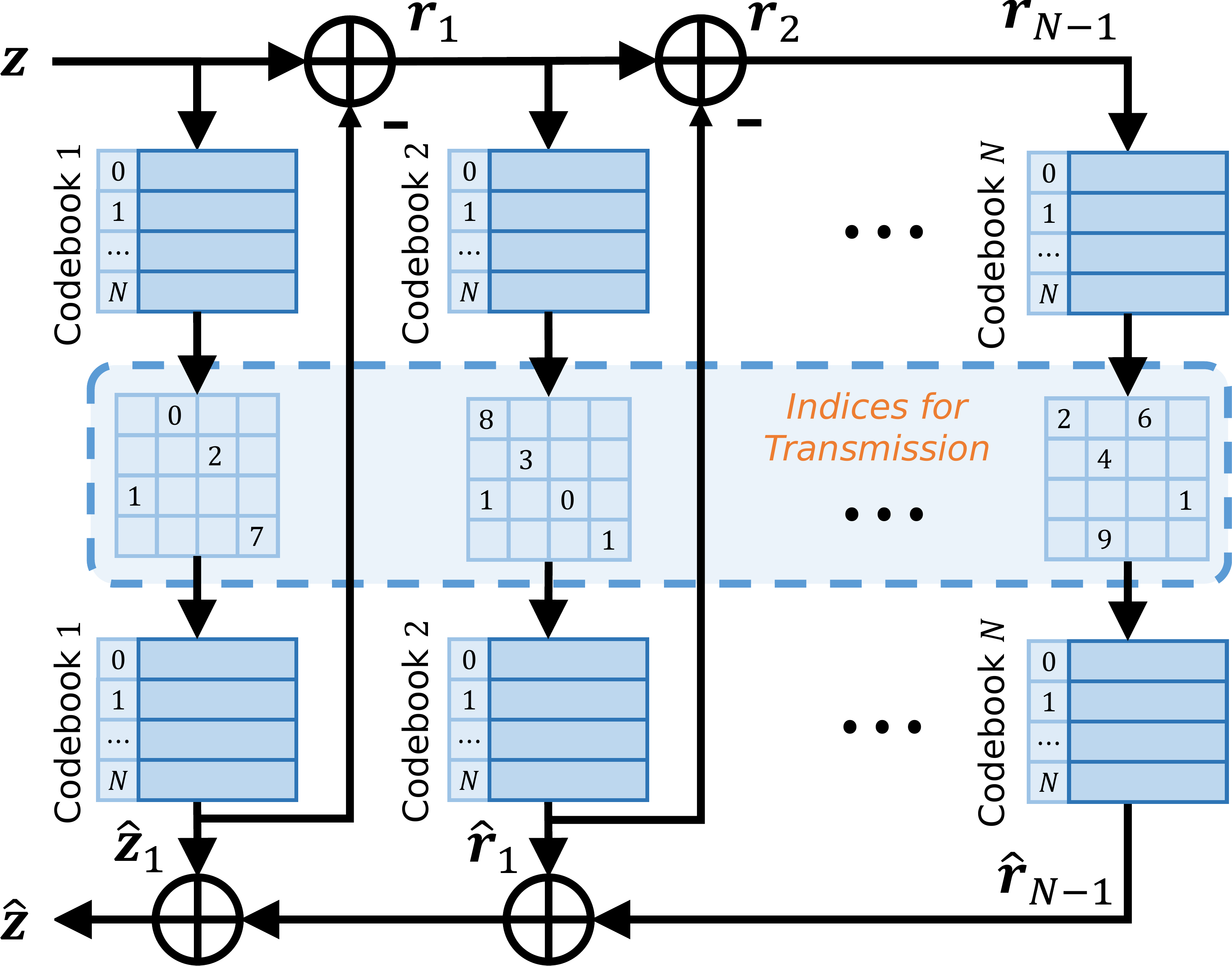}
      \caption{Residual vector quantizer (RVQ).}
      \label{fig:short-b}
    \end{subfigure}
    \vbox to 1em{}  
    \vfill
    \begin{subfigure}{0.98\linewidth}
      \includegraphics[width=0.99\linewidth]{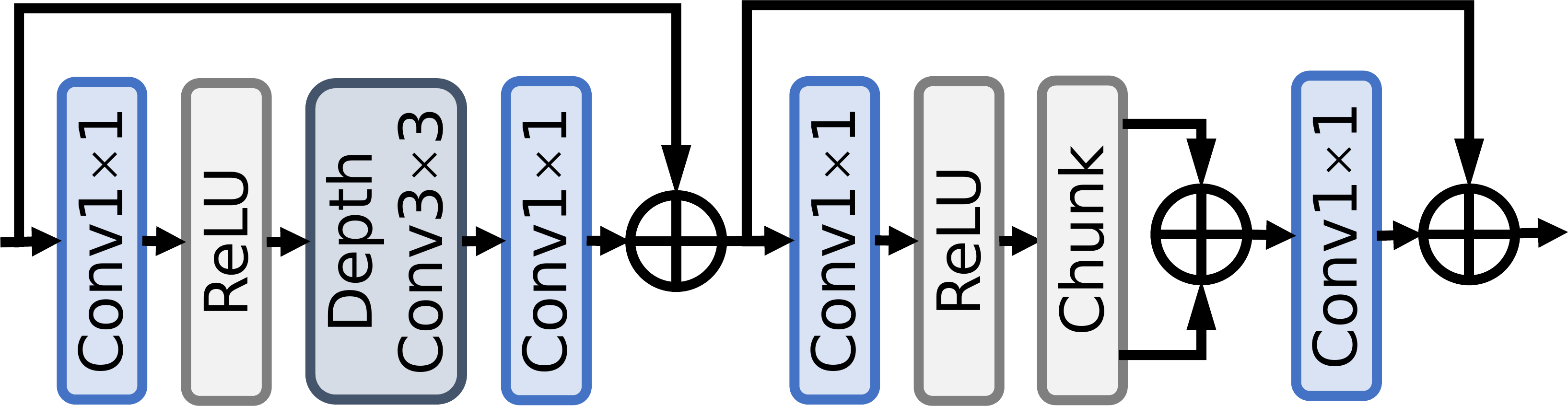}
      \caption{DC Block.}
      \label{fig:short-c}
    \end{subfigure}
  \end{subfigure}
  \caption{(a) Implementation details of EF-LIC, which largely follow DCVC-RT \cite{dcvc-rt}. The quantizer is realized as a set of RVQ modules with different numbers of codebooks, denoted by $m$. A rate-selection key determines which quantizer is used for a given inference. (b) RVQ architecture, following \cite{speech-dac}. (c) DC block architecture, following \cite{dcvc-rt}.}
  \label{fig:short2}
\end{figure*}

It remains to verify that the constructed indices satisfy the entropy-gap bound \eqref{eq:entropy_gap}.
Since $\hat Y_{<i}^{\mathrm{RD}}=\hat Y_{<i}^{\mathrm{CM}}=Z_i$, we have
\begin{equation}
H(J_i^{\mathrm{RD}}\mid \hat Y_{<i}^{\mathrm{RD}})
=
H(J_i^{\mathrm{RD}}\mid Z_i).
\end{equation}
The map $(S_i,B_i)\mapsto J_i^{\mathrm{RD}}=\iota_i(S_i,B_i)$ is injective.
Injective re-encodings preserve conditional entropy, so
\begin{equation}
H(J_i^{\mathrm{RD}}\mid Z_i)=H(S_i,B_i\mid Z_i).
\end{equation}
Using the identity for $H(S_i,B_i\mid Z_i)$ yields
\begin{equation}
H(J_i^{\mathrm{RD}}\mid \hat Y_{<i}^{\mathrm{RD}})
=
H(S_i\mid Z_i)+\log M_i.
\end{equation}
Summing over $i$ gives
\begin{equation}
\sum_{i=1}^N H(J_i^{\mathrm{RD}}\mid \hat Y_{<i}^{\mathrm{RD}})
=
R_0+\sum_{i=1}^N \log M_i
\ge
R.
\end{equation}
Using $\sum_{i=1}^N n_i\log K_i=R'$ and the definition of $\Delta \bar H$ in \eqref{eq:entropy_gap}, we obtain
\begin{equation}
\begin{aligned}
\Delta \bar H
&=
\frac{\sum_{i=1}^N \left(n_i\log K_i-H(J_i^{\mathrm{RD}}\mid \hat Y_{<i}^{\mathrm{RD}})\right)}{\sum_{i=1}^N n_i\log K_i} \\
&\le
\frac{R'-R}{R'} \\
&=
\varepsilon,
\end{aligned}
\end{equation}
up to the negligible rounding slack in the choices of $\{K_i\}$ and $\{M_i\}$.

Thus the constructed RD scheme is feasible at rate $R'$ and satisfies $\Delta\bar H\le\varepsilon$.
Since $D_X^{\mathrm{RD}}(R')$ is the infimum distortion over all such RD schemes, we conclude
\begin{equation}
D_X^{\mathrm{RD}}(R')
\le
\mathbb{E}\big[d(X,\hat X^{\mathrm{RD}})\big]
=
D_X^{\mathrm{CM}}(R).
\end{equation}
Substituting $R'=R/(1-\varepsilon)$ yields
\begin{equation}
D_X^{\mathrm{RD}}\!\left(\frac{R}{1-\varepsilon}\right)\le D_X^{\mathrm{CM}}(R).
\end{equation}
Finally, since $\varepsilon\in(0,1)$ was arbitrary, letting $\varepsilon\downarrow 0$ shows that the rate overhead can be made arbitrarily small.
\end{proof}

\section{Detailed Model Architectures}
\label{app:architecture}
In this section, we provide additional implementation details of EF-LIC.
As shown in \cref{fig:short2}, the EF-LIC backbone is composed of DC Blocks, which implement depthwise separable convolutions following \cite{dcvc-rt}.
In this architecture, \texttt{Patch} denotes a pixel-unshuffle operation with an upscaling factor of 8, and \texttt{Merge} denotes the inverse operation.
We set $C_1=368$, $C_y=256$, and $C_z=128$.
We implement RVQ following \cite{speech-dac}, as illustrated in \cref{fig:short-b}.
To support multiple bitrates, we use a set of independent RVQ modules, where each RVQ uses a different number of codebooks $m$.
In the main text, we set $m \in \{1,2,3,4,5\}$ to cover a sufficiently wide bitrate range.
At inference time, in addition to the input image, the model takes a rate-selection parameter $q$ that determines which RVQ is used for quantization.

As shown in \cref{fig:short-b}, we transmit the quantized indices produced by each RVQ codebook.
Each index tensor $j$ has shape $1 \times h \times w$.
We flatten the indices from all codebooks into a one-dimensional vector.
Within each RVQ, we concatenate the flattened indices in codebook order.
We then concatenate the RVQ vectors in the order $Q_{\bm{z}} \rightarrow Q_1 \rightarrow Q_2 \rightarrow Q_3 \rightarrow Q_4  $.
For transmission, we prepend a header containing $H$, $W$, and $q$, where $H \times W$ is the input image resolution and $q$ is the rate-selection parameter.
The header takes 28 bits for $H$ and $W$ and 4 bits for $q$, which is negligible compared to the overall bitrate.
Given a fixed model, the mapping from $H \times W$ to the index grid $h \times w$ is deterministic, and the number of codebooks and codewords in each RVQ is fixed.
Therefore, these header fields are sufficient to parse the stream and recover all RVQ indices.
Notably, our index packing introduces no sequential dependency and requires no expensive operations beyond concatenation in a predefined order.
As a result, both encoding and decoding are highly efficient and take less than 1\,ms in total in our implementation.
Furthermore, the bit packing and unpacking of the previous quantizer are independent of the computation of the subsequent quantizer, so the two can be overlapped in parallel, making the end-to-end latency of this step nearly negligible.


\section{Experimental Details}\label{app:sec3}
\subsection{Performance Details}
This section provides additional details on the baselines described in the main text.
For H.266/VVC (VTM)~\cite{vtm}, we adopt its intra-only coding configuration, which is among the strongest engineered baselines for still-image compression.
We evaluate VTM v23.10 to reflect contemporary encoder and decoder runtimes and to enable a fair speed comparison.
We compile VTM on Linux and run intra coding with the following command:
\begin{verbatim}
EncoderApp
  -i [input.yuv]
  -c encoder_intra_vtm.cfg
  -o [output.yuv] 
  -b [output.bin]
  --wdt [width] 
  --hgt [height] 
  -q [QP]
  --InputBitDepth=8 
  -fr 1 
  -f 1
  --InputChromaFormat=420
\end{verbatim}
We use YUV420-formatted inputs, as this chroma subsampling setting yields faster runtimes.
For Control-GIC~\cite{cgic}, we exhaustively search over all granularity combinations using a step size of $0.01$ and report the best-performing configuration.
We observe substantial quality degradation for Control-GIC when the BPP falls below 0.15.
Following the protocol in the original paper, we restrict BD-rate computation to the range $\text{BPP} \ge 0.15$.
In addition, we find that the encoding and decoding runtime of Control-GIC grows approximately quadratically with the number of pixels, whereas the other models scale approximately linearly.
At a resolution of $256 \times 256$, our measured encoding and decoding times closely match those reported in the original paper.
At the standard Kodak resolution of $512 \times 768$, however, Control-GIC becomes substantially slower.
On DIV2K and CLIC2020, we use the official tiling function to prevent out-of-memory errors.
For OSCAR~\cite{oscar}, we evaluate the author-released code and pretrained models.
The official implementation, however, does not support high-resolution image evaluation.
So we offload selected model components to CPU memory during inference.
For RDEIC~\cite{rdeic}, we use the checkpoint at step 2.
The official implementation of HiFiC~\cite{hific} depends on an older TensorFlow release and does not run on recent GPUs such as the NVIDIA A100 or RTX 5090.
For comparability, we instead use a community PyTorch reimplementation together with its released pretrained weights.
For all other baselines, we use the official implementations and pretrained checkpoints.

We compute LPIPS~\cite{lpips} with the \texttt{lpips} Python package, normalizing inputs to $[-1, 1]$ as in the official setup and using pretrained VGG~\cite{vgg} weights, which are commonly adopted for LPIPS-based visual-quality evaluation.
We compute DISTS~\cite{dists} with \texttt{DISTS\_pytorch} and normalize inputs to $[0, 1]$.
We measure FLOPs with the \texttt{calflops} Python library and follow the convention $1~\text{FLOP}=2~\text{MACs}$.
Bj{\o}ntegaard delta rate (BD-rate)~\cite{bd-rate} measures the average bitrate difference between two methods over a specified quality range.
We compute BD-rate as the area between the two R--D curves after interpolating them with a monotonic piecewise cubic Hermite interpolating polynomial (PCHIP).
A negative BD-rate indicates that the proposed method achieves the same quality at a lower bitrate than the baseline.
We use the \texttt{bjontegaard} Python library to perform these calculations.

\subsection{Ablation Details}
\label{app:exp_detail}

In this section, we provide additional training details for the ``UQ+EC'' and ``VQ+EC'' models in \cref{sec:ablation}.
For typical LIC ``UQ+EC'', we optimize an objective that includes an explicit rate term $R$ weighted by a Lagrange multiplier $\lambda$. The training loss is defined as
\begin{equation}
\mathcal{L} = D + \lambda R,
\end{equation}
where
\begin{equation}
    D = \lVert \bm{x} - \hat{\bm{x}}_m\rVert_1
      + \lambda_{\mathrm{per}} \, \mathcal{L}_{\mathrm{per}}(\bm{x}, \hat{\bm{x}}_m) + \lambda_{\mathrm{adv}} \, \mathcal{L}_{\mathrm{adv}}(\bm{x}, \hat{\bm{x}}_m),
\end{equation}
which is consistent with \cref{eq:loss} for EF-LIC. $R$ is the expected bitrate estimated by the context model, where
\begin{equation}
  \label{eq:rate_loss}
  R = \sum_{i=1}^N
  \mathbb{E}\!\left[
    -\log P_{\hat Y_i\mid \hat Y_{<i}}(\hat Y_i \mid \hat Y_{<i};\theta_i)
  \right].
\end{equation}

The Lagrange multiplier $\lambda$ controls the resulting bitrate.
We train models with $\lambda \in \{0.5, 0.75, 1.0, 1.5, 2.0\}$ to span a bitrate range comparable to that of EF-LIC.
The resulting R--D curves are shown in \cref{fig:ablation}, where EF-LIC, ``UQ+EC'', ``VQ+EC'', and the VQ baseline cover similar bitrate ranges.

For ``VQ+EC'', VQ blocks gradients to the context model.
We therefore add an explicit rate loss term $R$ consistent with \cref{eq:rate_loss} to the objective in \cref{eq:loss}.

\begin{figure}[t]
  \centering
  \includegraphics[width=0.95\linewidth]{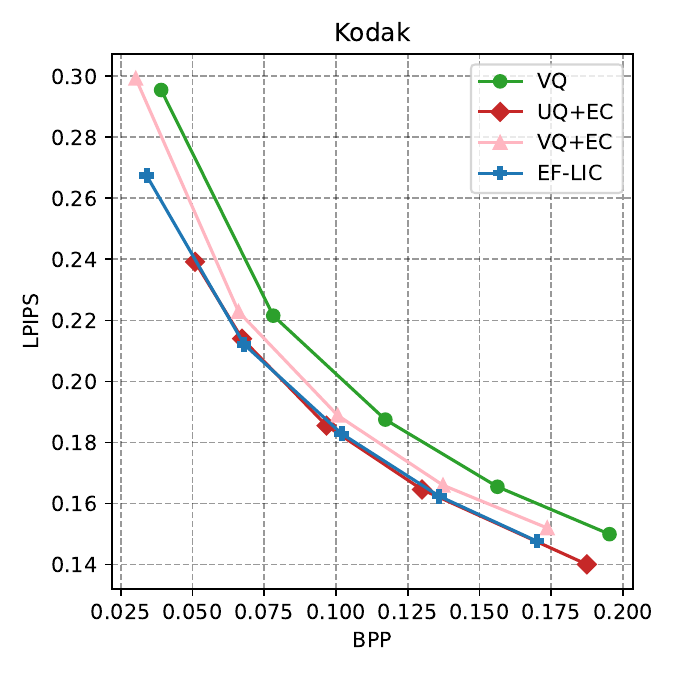}
  \caption{R-D performance on the Kodak dataset, evaluated with LPIPS vs. BPP. Curves closer to the lower-left  are better.}
  \label{fig:ablation}
\end{figure}

\section{Extra Experimental Results}\label{app:sec4}
\subsection{Quantitative Results for Entropy Gap}\label{app:entropy}
In \cref{fig:entropy}, we report $\Delta H$ for each codebook on the Kodak dataset when RVQ uses five codebooks.
Using \cref{eq:entropy_gap}, we obtain $\Delta \bar H = 3.42\%$, which is consistent with the conclusions in \cref{thm:distortion_entropy} and \cref{eq:pdf_vq}.
In addition, the quantizer for the latents $\bm{y}$ exhibits high codebook utilization, whereas the hyperprior quantizer $Q_{\bm{z}}$ for $\bm{z}$ shows low utilization.
This suggests that, while performing decorrelation in the representation domain, the method also regularizes the latent distribution, making it easier for VQ to learn the probability shaping.

\begin{figure}[t]
  \centering
  \includegraphics[width=0.95\linewidth]{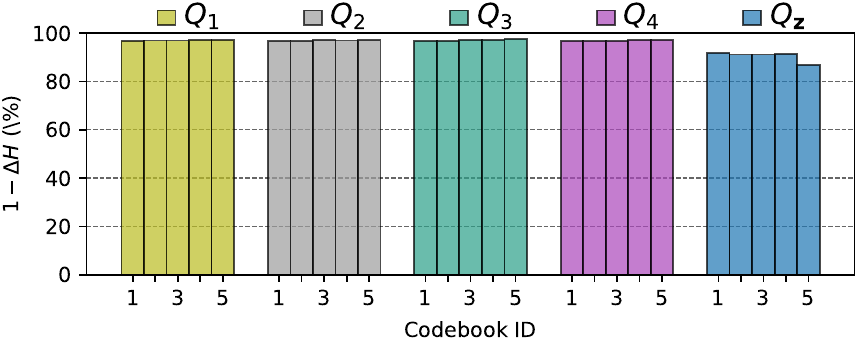}
  \caption{Normalized codebook entropy for each codebook in $Q_{1}$--$Q_{4}$ and $Q_{\bm{z}}$, where there are 5 codebooks in each RVQ. Each bar reports $1-\Delta H$ for the corresponding quantizer. A higher bar denotes less statistical redundancy.}
  \label{fig:entropy}
\end{figure}

\begin{table}[t]
  \caption{Comparison of BD-rate on the Kodak, Tecnick, DIV2K, and CLIC 2020 datasets evaluated under DISTS. Best results are in \textbf{bold}. Second-best are \underline{underlined}.}
  \label{tab:dists}
  \centering
  \begingroup
  \setlength{\tabcolsep}{4.0pt} 
  \begin{tabular}{@{}lrrrr@{}}
    \toprule
    \multirow{2}{*}{Method} &
    \multicolumn{4}{c}{BD-rate (DISTS)} \\
    \cmidrule(lr){2-5}
    & Kodak & Tecnick & DIV2K & CLIC2020\\
    \midrule
    HiFiC & 90.08\%   & 99.67\%   & 100.76\%   & 124.45\%  \\
    Control-GIC & 34.18\%   & 67.12\%   & 62.09\%   & 110.76\%  \\
    MS-ILLM  & 0.00\% & 0.00\% & 0.00\% & 0.00\% \\
    DiffEIC  & -33.79\%  & 23.68\%   & 15.78\%        & 59.91\%        \\
    OSCAR & -50.63\%  & -4.76\%    & -20.57\%       & 36.02\%          \\
    RDEIC & -39.92\%  & -3.01\%   & -4.46\%        & 13.67\%        \\
    \midrule
    EF-LIC-s & \underline{-58.61\%}  & \underline{-18.75\%}    & \underline{-45.30\%}  & \underline{-10.62\%}  \\
    EF-LIC  & \textbf{-70.61\%}  & \textbf{-43.12\%}   & \textbf{-60.43\%}  & \textbf{-42.15\%}  \\
    \bottomrule
  \end{tabular}
  \endgroup
\end{table}

\subsection{Quantitative Results on DISTS}
\label{app:dists}
In \cref{fig:RD-curve}, we have presented the R--D curves of EF-LIC and the baseline methods on multiple datasets, measured using DISTS \cite{dists}.
In this section, we further report quantitative BD-rate comparisons under DISTS, as summarized in \cref{tab:dists}.
EF-LIC significantly outperforms the baseline methods evaluated under DISTS as well.
Moreover, EF-LIC and EF-LIC-s are the only methods that achieve better DISTS performance than MS-ILLM on every dataset, especially CLIC 2020.

\subsection{Additional Comparison with Recent Generative Codecs}

We supplement an additional comparison with recent generative image compression methods, including GLC~\cite{vq-gic}, DLF~\cite{vq-DLF}, StableCodec~\cite{diffusion-StableCodec}, and OneDC~\cite{diffusion-xue2025one}.
The results are reported on Kodak and summarized in \cref{tab:extra_generative_comparison}. BD-rate is calculated with LPIPS, where OneDC is used as the anchor. EF-LIC achieves the best BD-rate while using substantially fewer FLOPs and parameters.

\begin{table*}[t]
  \caption{Additional comparison with recent generative image compression methods on Kodak measured with LPIPS. More negative BD-rate means lower bitrate at the same LPIPS. OneDC is used as the anchor. Best results are in \textbf{bold}. ``Enc./Dec.'' reports per-image encoding/decoding time.}
  \label{tab:extra_generative_comparison}
  \centering
  \begin{tabular}{@{}lrrrrr@{}}
    \toprule
    Method & Enc.(ms) & Dec.(ms) & FLOPs(G) & Params(M) & BD-rate \\
    \midrule
    GLC~\cite{vq-gic} 
      & 34.75  & 48.26  & 2466.28 & 163.99  & 21.50\% \\
    DLF~\cite{vq-DLF} 
      & 189.99 & 247.30 & 5809.71 & 1286.01 & 20.63\% \\
    StableCodec~\cite{diffusion-StableCodec} 
      & 100.18 & 200.03 & 6820.28 & 1065.81 & 8.03\% \\
    OneDC~\cite{diffusion-xue2025one} 
      & 100.50 & 235.03 & 7142.91 & 1406.42 & 0.00\% \\
    \midrule
    EF-LIC (Ours) 
      & \textbf{17.62} & \textbf{13.72} & \textbf{279.61} & \textbf{35.74} & \textbf{-3.33\%} \\
    \bottomrule
  \end{tabular}
\end{table*}

\begin{table*}[t]
  \caption{Comparison of GPU runtimes (ms) and memory (GB) for image encoding and decoding across different resolutions. Enc./Dec. denote encoding/decoding times. Mem. denotes memory usage. Best results are in \textbf{bold}. Second-best are \underline{underlined}.}
  \label{tab:high_resolution_time}
  \centering
  \begingroup
  \setlength{\tabcolsep}{3.5pt} 
  \begin{tabular}{@{}lrrrrrrrrrrrr@{}}
    \toprule
    \multirow{2}{*}{Method} & \multicolumn{3}{c}{512$\times$768} & \multicolumn{3}{c}{1080$\times$1920} & \multicolumn{3}{c}{1440$\times$2560} & \multicolumn{3}{c}{2160$\times$3840}\\
    \cmidrule(lr){2-4} \cmidrule(lr){5-7} \cmidrule(lr){8-10} \cmidrule(lr){11-13}
     & Enc. & Dec. & Mem. & Enc. & Dec. & Mem. & Enc. & Dec. & Mem. & Enc. & Dec. & Mem. \\
    \midrule
    H.266/VVC & $>$10s & 150.30 & -- & $>$10s & 230.10 & -- & $>$10s & 288.16 & -- & $>$10s & 486.71 & --  \\
    LIC-HPCM & 62.37 & 82.88 & 0.53 & 309.80 & 342.95 & 1.98 & 465.91 & 474.49 & 2.84 & 1121.92 & 1147.79 & 6.05  \\
    DCVC-RT & \underline{14.09} & 17.08 & 0.34 & 76.68 & 59.87 & 1.04 & 135.87 & 102.95 & 1.73 & 259.86 & 197.83 & 3.63  \\
    \midrule
    HiFiC & 526.51  & 1408.60 & 1.14 & 2894.55 & 6909.92 & 2.97 & 5179.44 & $>$10s & 4.78 & $>$10s & $>$10s & 9.75  \\
    Control-GIC & 103.56  & 436.26 & 6.53 & 610.76 & 2186.30 & 69.99 & -- & -- & $>$80 & -- & $>$10s & $>$80  \\
    MS-ILLM & 165.38  & 147.79 & 1.12 & 350.85 & 379.01 & 2.99 & 516.47 & 601.18 & 4.87 & 1305.93 & 1613.82 & 9.94  \\
    DiffEIC & 210.18 & 4661.74 & 6.86 & -- & $>$10s & -- & -- & $>$10s & -- & -- & $>$10s & --  \\
    OSCAR & 53.04 & 167.56 & 5.57 & 513.20 & 1123.38 & 24.75 & 2382.55 & 2719.79 & 79.44 & -- & -- & $>$80 \\
    RDEIC & 157.25  & 426.68 & 6.86 &683.20&2296.58& 13.55 &1436.00&5963.77& 20.17 & --&$>$10s& -- \\
    \midrule
    EF-LIC & 17.62 & \underline{13.72} & \underline{0.25} & \underline{19.53} & \underline{35.62} & \underline{0.65} & \underline{31.58} & \underline{55.50} & \underline{1.04} & \underline{65.24} & \underline{116.28} & \underline{2.10}  \\
    EF-LIC-s & \textbf{9.94} & \textbf{6.26} & \textbf{0.15} & \textbf{11.66} & \textbf{14.85} & \textbf{0.55} & \textbf{14.39} & \textbf{22.77} & \textbf{0.94} & \textbf{30.18} & \textbf{45.31} & \textbf{2.00}  \\
    \bottomrule
    
  \end{tabular}
  \endgroup
\end{table*}

\subsection{Runtime Analysis on High-Resolution Images}
\label{app:runtime}

In this section, we report the encoding and decoding time, together with the peak GPU memory usage, of different methods at resolutions of $512\times768$, 1080p, 2K, and 4K.
We use the same hardware and experimental settings as in the main paper.
The results are summarized in Table~\ref{tab:high_resolution_time}.

Although EF-LIC already shows a substantial advantage at $512\times768$ as reported in the main text, this margin further increases as the resolution grows.
At 4K resolution, EF-LIC achieves a decoding speed close to $15\times$ that of MS-ILLM.
Moreover, when the resolution increases from $512\times768$ to 1080p, the encoding time of EF-LIC and EF-LIC-s changes only slightly.
This is because the RVQ nearest neighbor search has a low complexity on GPU, so increasing the resolution has little impact on its runtime. 
While the remaining convolutional modules scale approximately as $O(n)$ ($n$ denotes the number of pixels), which makes them become the latency bottleneck for compressing high resolution images.
At lower resolutions, RVQ accounts for most of the encoding time, but as the resolution increases, the convolutional components gradually become the dominant cost, which results in a relatively small increase in the overall encoding time.
This also explains why EF-LIC exhibits larger speed advantages on higher resolution images.

\begin{figure*}[t]
  \centering
  \includegraphics[width=0.97\linewidth]{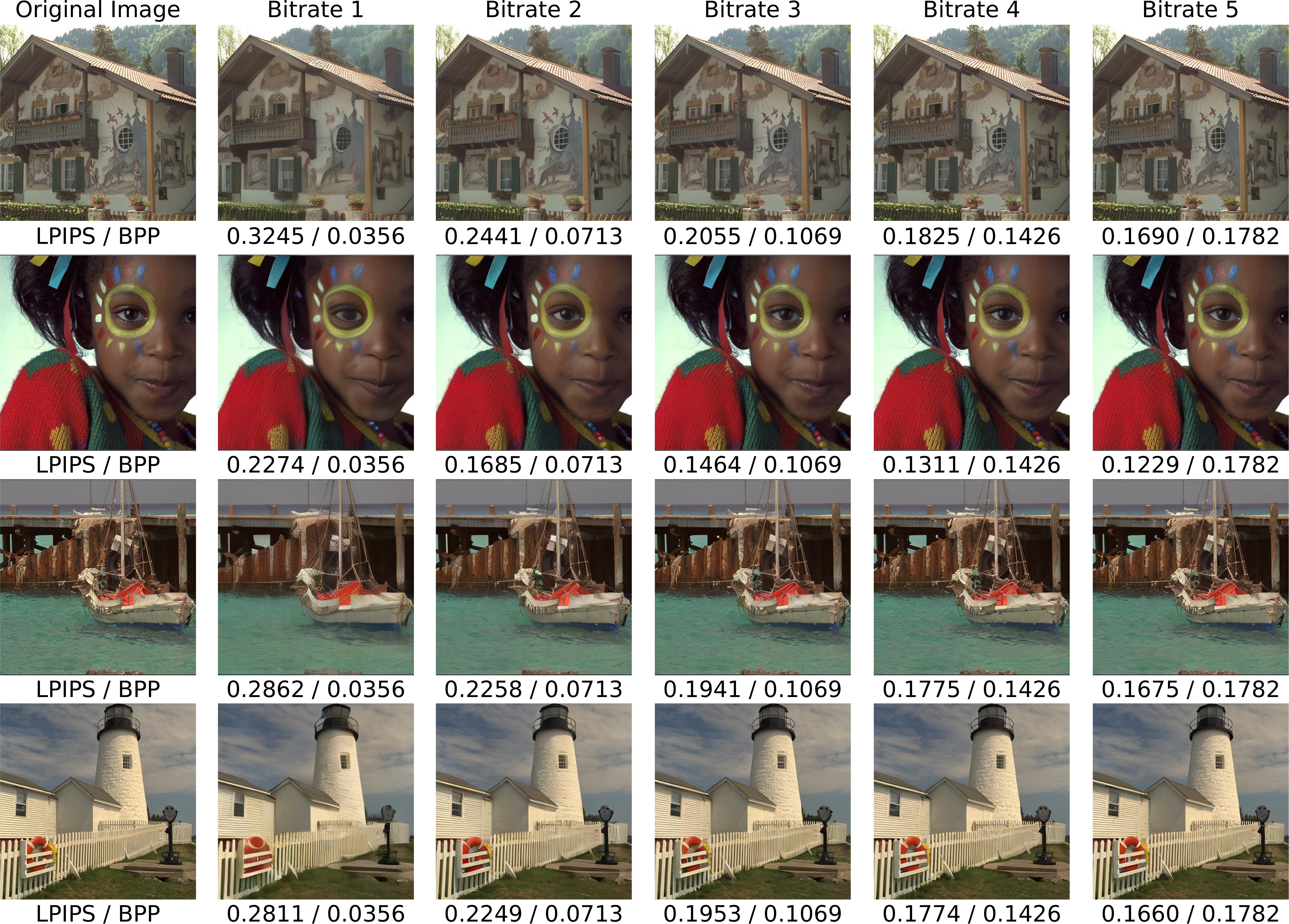}
  \caption{Qualitative results of EF-LIC at different bitrates on Kodak.  The bitrate increases from left to right.}
  \label{fig:vis_1}
\end{figure*}

\subsection{Comparison with Advanced Entropy Coding}
In the main paper, we use the rANS \cite{rans} implementation provided by CompressAI \cite{compressai} because it has been widely adopted in most LIC \cite{lic-balle2018variational,lic-joint,lic-elic,lalic,hpcm}.
We compare against the stronger entropy coding implementations in DCVC-RT \cite{dcvc-rt}.
The results are also included in \cref{tab:high_resolution_time}.
At $512\times768$, DCVC-RT is slightly faster than EF-LIC in encoding.
However, their entropy coder remains an $O(n)$ operation, and its runtime increases substantially as the resolution grows.
Consequently, EF-LIC becomes notably faster than DCVC-RT at 1080p, and the advantage further widens at 4K.

\subsection{Quantitative Results for Other Metrics}
\label{app:others}

Although we report the performance of EF-LIC under LPIPS~\cite{lpips} and DISTS~\cite{dists} in the main paper, we provide R--D curves measured by PSNR, MS-SSIM~\cite{ms-ssim}, FID~\cite{fid}, KID~\cite{kid}, NIQE~\cite{niqe}, MUSIQ~\cite{musiq}, and CLIP-IQA~\cite{clipiqa} in this section to verify that the LPIPS and DISTS improvements do not come at the cost of distortion-based quality.
As shown in \cref{fig:other_RD-curve} and \cref{fig:other_RD-curve2}, EF-LIC achieves comparable performance to the competing methods under these metrics.

Moreover, we do not emphasize FID~\cite{fid} in the main paper because we find that, while FID reflects the realism of generated images, it does not directly measure the similarity between a reconstruction and its corresponding source image.
As illustrated in \cref{fig:vis}, methods based on Stable Diffusion~\cite{stable-diffusion} can produce visually realistic images, but their content can differ substantially from the original images, which leads to a large FID in our evaluation.
Since our goal is image compression rather than image generation, preserving fidelity to the original content is essential, and we therefore primarily report LPIPS and DISTS in the main paper.

Following GLC~\cite{vq-gic}, we adopt their evaluation methodology for FID~\cite{fid} and KID~\cite{kid}. This protocol crops images into $256 \times 256$ non-overlapping patches to significantly augment the sample size, thereby ensuring a more accurate and robust calculation of both metrics. While this approach aligns more closely with mainstream evaluation paradigms in recent works \cite{vq-gic,vq-DLF,diffusion-StableCodec,diffusion-xue2025one}, it deviates from the configuration we previously reported in the rebuttal.

\begin{figure*}[t]
  \centering
  \includegraphics[width=0.995\linewidth]{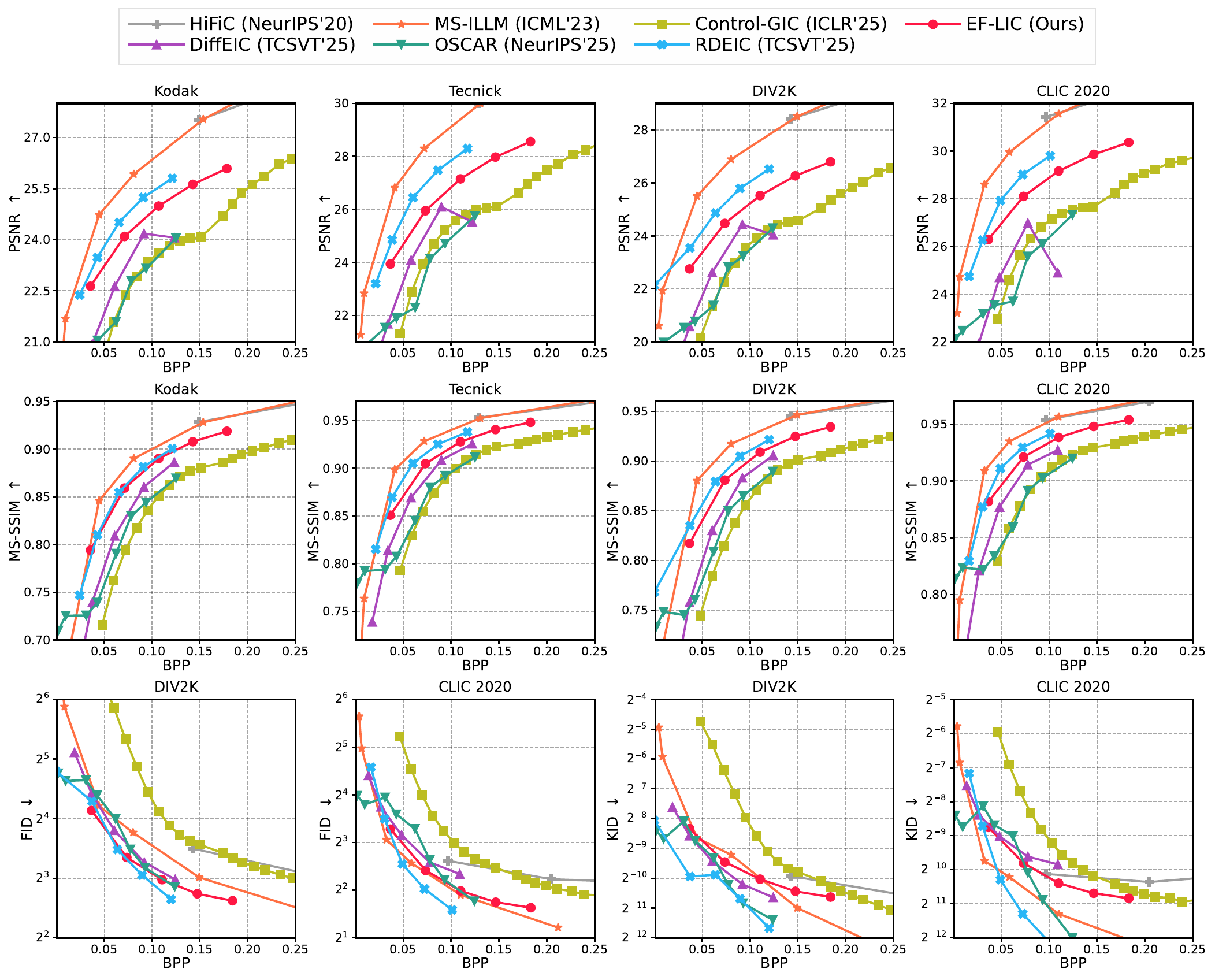}
  \caption{R--D performance on the Kodak, Tecnick, DIV2K, and CLIC2020 datasets, evaluated with PSNR, MS-SSIM, FID and KID vs. BPP.}
  \label{fig:other_RD-curve}
\end{figure*}

\begin{figure*}[t]
  \centering
  \includegraphics[width=0.995\linewidth]{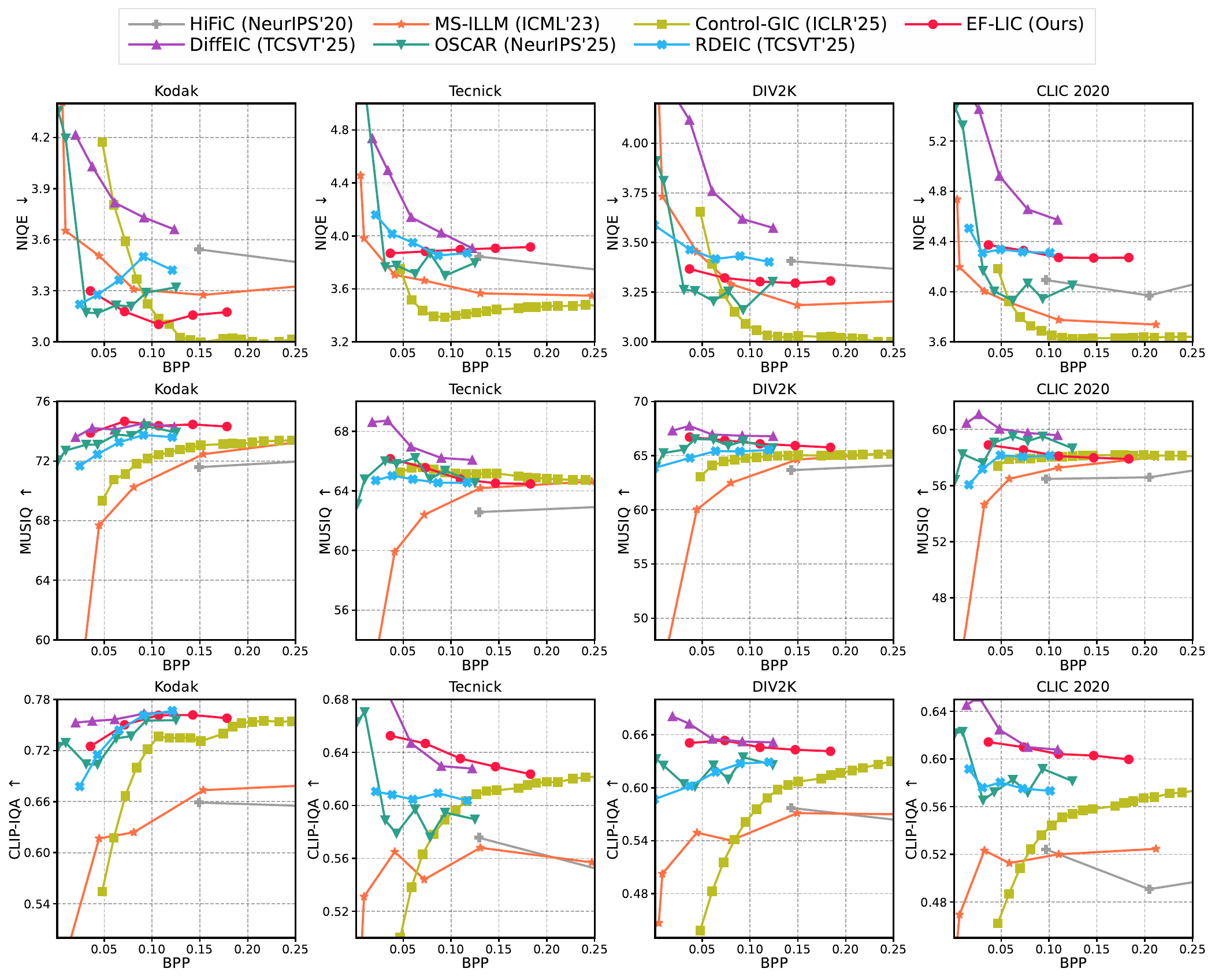}
  \caption{R--D performance on the Kodak, Tecnick, DIV2K, and CLIC2020 datasets, evaluated with NIQE, MUSIQ and CLIP-IQA vs. BPP.}
  \label{fig:other_RD-curve2}
\end{figure*}

\subsection{More Visualization Results}

In this section, we provide additional visualization results of EF-LIC.
\cref{fig:vis_1} presents qualitative results of EF-LIC at different bitrates, showing that EF-LIC effectively supports multi-rate compression.
We further present qualitative results of EF-LIC on the high-resolution Tecnick \cite{tecnick}, DIV2K \cite{div2k}, and CLIC 2020 \cite{clic2020} datasets in \cref{fig:vis_3,fig:vis_4,fig:vis_2,fig:vis_5}.
Although the qualitative results of different models on high-resolution images appear similar, we include them to demonstrate that EF-LIC also functions correctly on high-resolution inputs.

\begin{figure*}[t]
  \centering
  \includegraphics[width=0.85\linewidth]{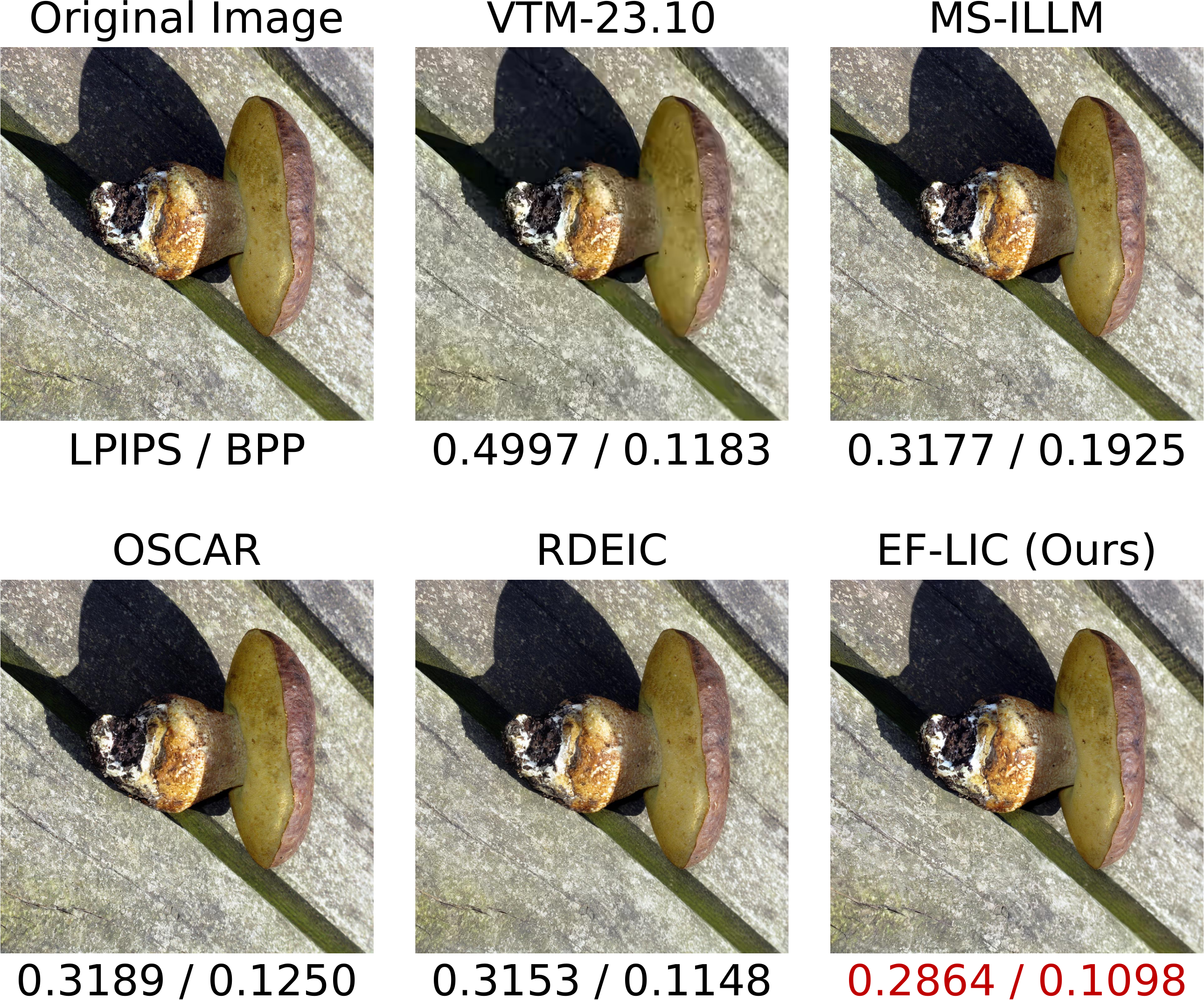}
  \caption{Visual comparison on Tecnick \cite{tecnick}. Numbers are LPIPS/BPP. Lower values indicate better visual quality and higher compression.}
  \label{fig:vis_3}
\end{figure*}

\begin{figure*}[t]
  \centering
  \includegraphics[width=0.90\linewidth]{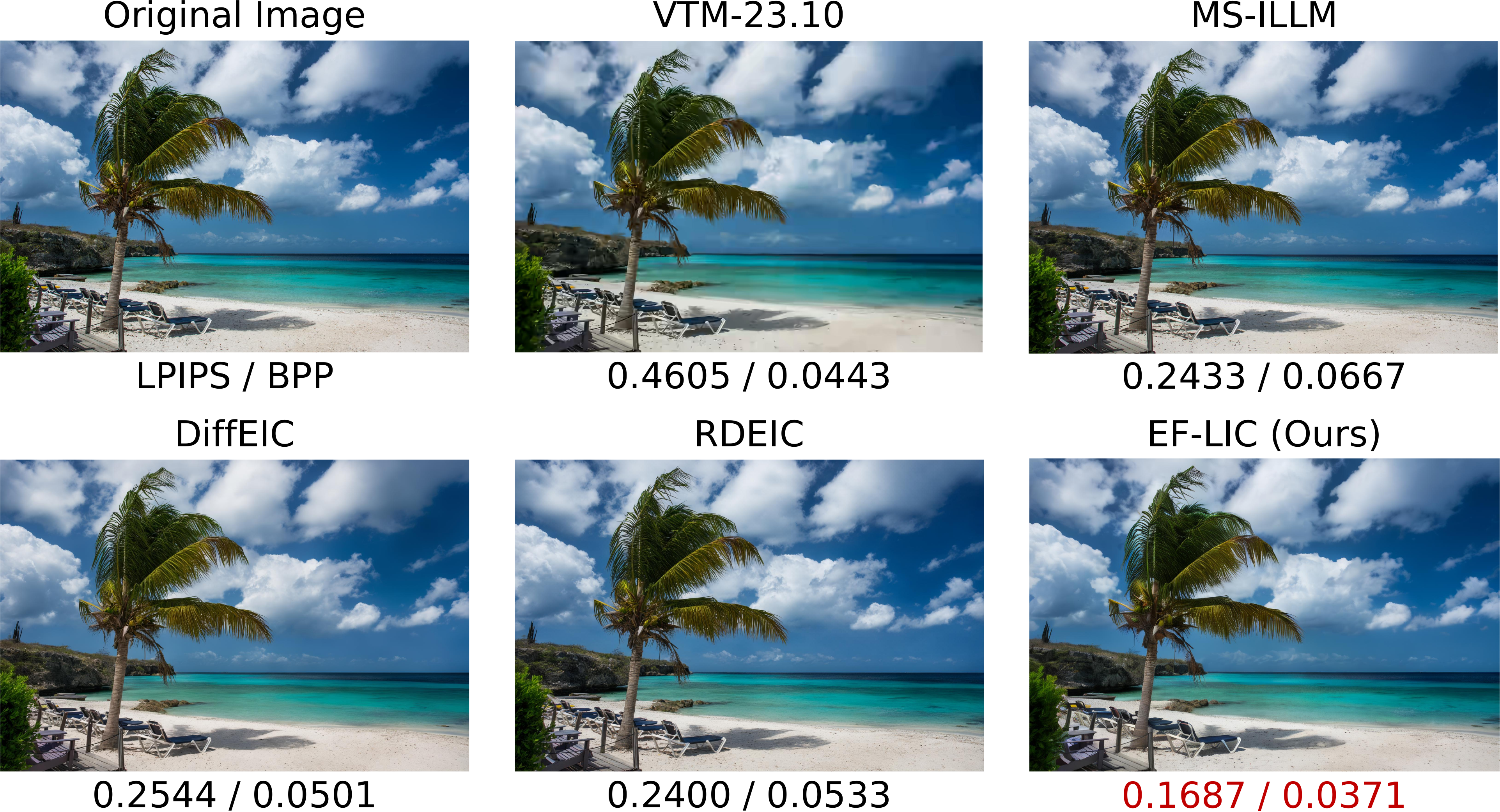}
  \caption{Visual comparison on DIV2K \cite{div2k}. Numbers are LPIPS/BPP. Lower values indicate better visual quality and higher compression.}
  \label{fig:vis_4}
\end{figure*}

\begin{figure*}[t]
  \centering
  \includegraphics[width=0.85\linewidth]{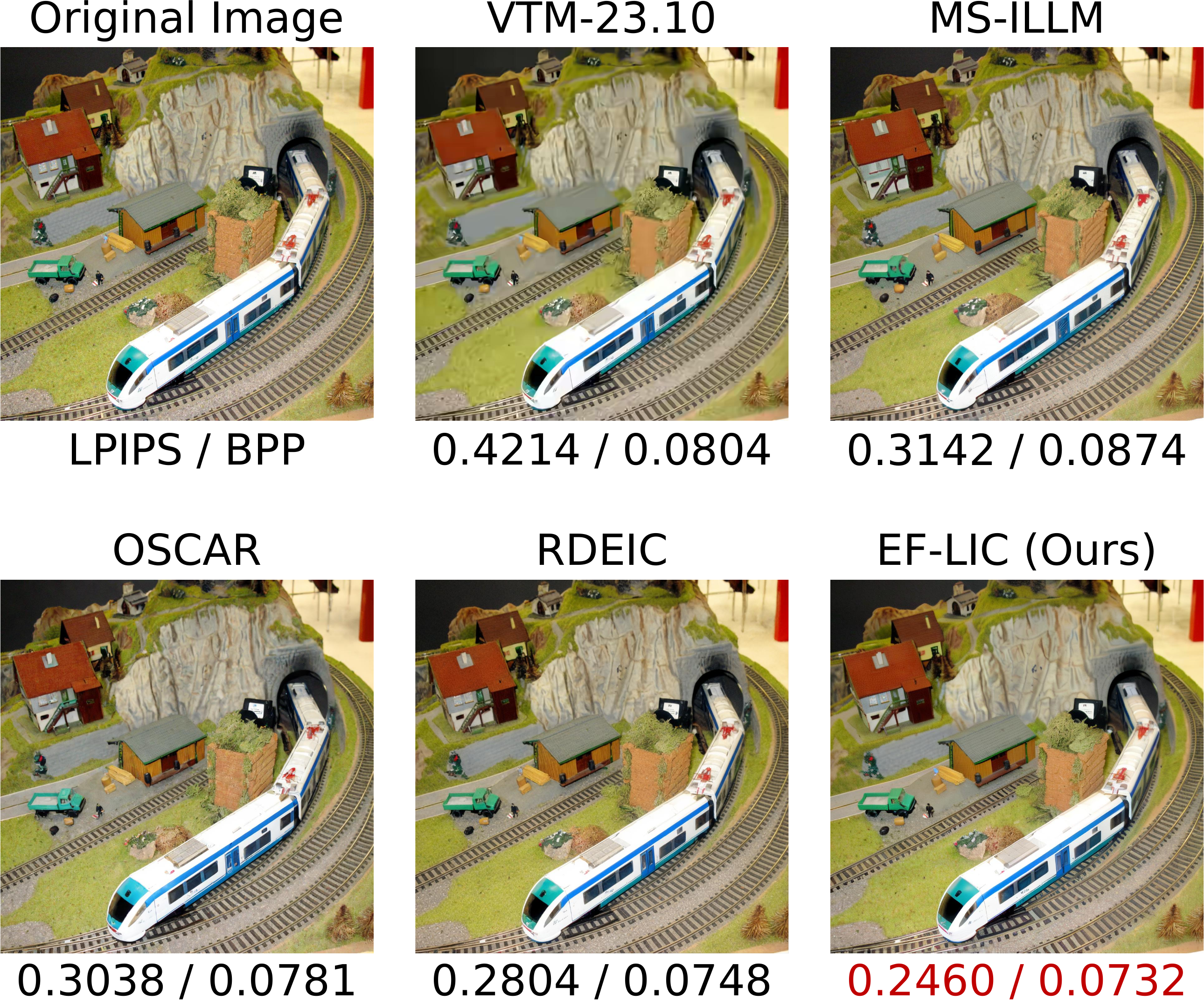}
  \caption{Visual comparison on Tecnick \cite{tecnick}. Numbers are LPIPS/BPP. Lower values indicate better visual quality and higher compression.}
  \label{fig:vis_2}
\end{figure*}

\begin{figure*}[t]
  \centering
  \includegraphics[width=0.95\linewidth]{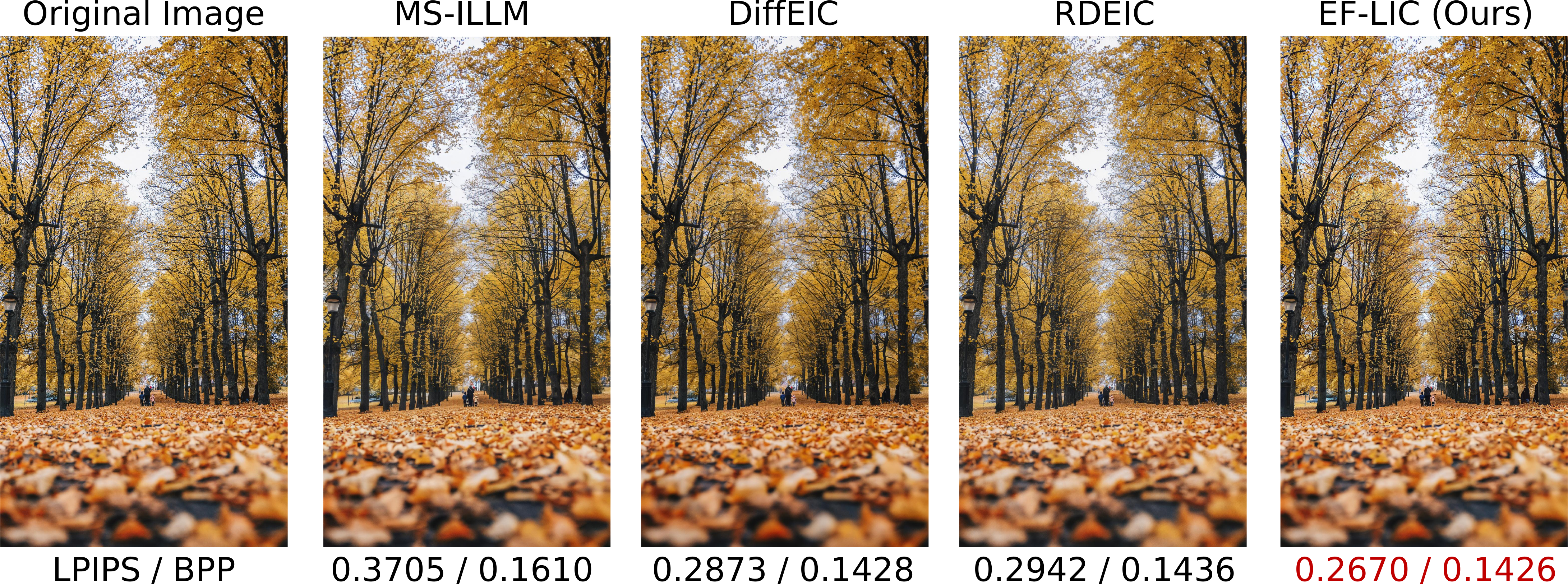}
  \caption{Visual comparison on CLIC 2020 \cite{clic2020}. Numbers are LPIPS/BPP. Lower values indicate better visual quality and higher compression.}
  \label{fig:vis_5}
\end{figure*}

\end{document}